\documentclass[runningheads]{llncs}

\usepackage[T1]{fontenc}

\usepackage{amsmath,amssymb}
\usepackage{mathtools}
\usepackage{hyperref}
\usepackage{cleveref}
\crefname{figure}{Figure}{Figures}
\crefname{algorithm}{Algorithm}{Algorithms}
\crefname{appendix}{Appendix}{Appendix}
\crefname{lemma}{Lemma}{Lemma}
\usepackage{siunitx}
\usepackage{braket}
\usepackage{bcprules}
\usepackage{tikz}
\usetikzlibrary{calc,quantikz,automata,positioning,shapes.callouts}

\usepackage{wrapfig}

\usepackage{graphicx}
\usepackage{subcaption}
\usepackage{xcolor}

\crefname{section}{Section}{Section}
\crefname{lemma}{Lemma}{Lemma}

\NewDocumentCommand\pfcase{m}{
  \smallbreak
  \noindent\textbf{Case:\qquad #1}\\[\smallskipamount] \ignorespaces
}
\NewDocumentCommand\pfcaselabel{m}{
  \smallbreak
  \noindent\textbf{Case #1:}
}

\usepackage{listings}
\lstdefinelanguage{Qtile}{
  keywords=[0]{if,then,else,mkref,let,in,return},
  morecomment=[l]{//},
  morecomment=[s]{/*}{*/}
}
\definecolor{comment-green}{rgb}{0,0.6,0}
\usepackage{algorithm}
\usepackage[noend]{algpseudocode}

\newif\ifdraftComments
\draftCommentstrue
\def\mkDraftFn#1#2{%
  \expandafter\def\csname #1\endcsname##1{\ifdraftComments\textcolor{#2}{[#1: ##1]}\marginpar[$\longrightarrow$]{$\longleftarrow$}\fi}%
}
\mkDraftFn{RW}{red}

\NewDocumentCommand\etal{}{\textit{et al}.\ }
\NewDocumentCommand\langname{}{\mathcal{Q}_{LS}}

\NewDocumentCommand\GateCX{}{\mathit{CX}}
\NewDocumentCommand\doubleplus{}{\mathbin{+\mkern-10mu+}}

\DeclareMathOperator{\dom}{dom}
\DeclareMathOperator{\cod}{cod}
\DeclareMathOperator{\Tr}{Tr}
\DeclareMathOperator{\Used}{Used}

\NewDocumentCommand\evalarrow{o}{\ensuremath{
  \IfValueT{#1}{\rightarrow_{#1}}
  {\rightarrow_{D, G}} 
}}

\NewDocumentCommand\rstate{mmm}{\ensuremath{\left[ #1, #2, #3 \right]}}
\NewDocumentCommand\lblketbra{mmm}{\ensuremath{\ket{#1}_{#3}\bra{#2}}}

\NewDocumentCommand\logiQRaw{mmmmmm}{
  \fill[#5] (#1, #2) rectangle node[black]{#6} ++ (#3, #4);
  \draw[very thick] (#1, #2)--(#1+#3, #2)--(#1+#3, #2+#4)--(#1,#2+#4)--cycle;
}
\NewDocumentCommand\logiQSq{mmmmm}{
  \logiQRaw{#1}{#2}{#3}{#3}{#4}{#5}
}

\newcommand{\langnt}[1]{\mathit{#1}}
\newcommand{\langmv}[1]{\mathit{#1}}
\newcommand{\langkw}[1]{\mathbf{#1}}
\newcommand{\langsym}[1]{#1}

\newif\ifarxiv
\arxivtrue
\NewDocumentCommand{\forarxiv}{m}{\ifarxiv #1 \fi}
\NewDocumentCommand{\forconf}{m}{\ifarxiv\else #1 \fi}

\begin{document}


\title{Type-Based Verification of Connectivity Constraints in Lattice Surgery}

\author{Ryo Wakizaka\inst{1}\orcidID{0000-0001-8762-9335} \and
Yasunari Suzuki\inst{2} \and
Atsushi Igarashi\inst{1}\orcidID{0000-0002-5143-9764}}

\authorrunning{R. Wakizaka et al.}

\institute{
  Graduate School of Informatics, Kyoto University, Japan,\\ \email{wakizaka@fos.kuis.kyoto-u.ac.jp} \\
  \email{igarashi@kuis.kyoto-u.ac.jp}
  \and
  NTT Computer and Data Science Laboratories, Musashino 180-8585, Japan \email{yasunari.suzuki@ntt.com}
}

\maketitle

\begin{abstract}
Fault-tolerant quantum computation using lattice surgery can be abstracted as operations on graphs, wherein each logical qubit corresponds to a vertex of the graph, and multi-qubit measurements are accomplished by connecting the vertices with paths between them. Operations attempting to connect vertices without a valid path will result in abnormal termination. As the permissible paths may evolve during execution, it is necessary to statically verify that the execution of a quantum program can be completed.

This paper introduces a type-based method to statically verify that well-typed programs can be executed without encountering halts induced by surgery operations. Alongside, we present $\langname$, a first-order quantum programming language to formalize the execution model of surgery operations. Furthermore, we provide a type checking algorithm by reducing the type checking problem to the offline dynamic connectivity problem.

\keywords{fault-tolerant quantum computation \and lattice surgery \and program verification \and type systems}
\end{abstract}

\section{Introduction}\label{sec:intro}

Fault-tolerant quantum computation is a method that enables the large-scale quantum computation required for applications, for example, in quantum chemistry~\cite{babbush2018EncodingElectronic,lee2021EvenMore} and cryptanalysis~\cite{shor1997PolynomialTimeAlgorithms,gidney2021HowFactor,beverland2022AssessingRequirements}, by addressing quantum errors that occur on quantum hardware. In recent years, small-scale fault-tolerant quantum computation has begun to be realized on actual hardware~\cite{acharya2023SuppressingQuantum,krinner2022RealizingRepeated,erhard2021EntanglingLogical,bluvstein2024LogicalQuantum,dasilva2024DemonstrationLogical}, leading to an urgent need to develop quantum software such as quantum compilers to generate quantum programs executable on fault-tolerant quantum computers.

The primary technology to realize fault-tolerant quantum computation is quantum error-correcting codes, which build a logical qubit from multiple physical qubits. However, since physical qubits can only interact (i.e., perform multi-qubit operations) with neighboring qubits, error-correcting codes can be implemented in a way that satisfies such locality conditions. Numerous quantum error-correcting codes have been proposed, with topological codes~\cite{kitaev2003FaulttolerantQuantum} such as surface codes~\cite{fowler2012SurfaceCodes} and color codes~\cite{bombin2006TopologicalQuantum} standing out as promising candidates because they exhibit robust performance and are relatively straightforward to implement in quantum computers with locally connected physical qubits.

To achieve fault-tolerant quantum computation, it is imperative not only to construct logical qubits but also to execute logical operations on them. Several methods~\cite{fowler2012SurfaceCodes,horsman2012SurfaceCode} have been proposed to realize logical operations on topological codes, and lattice surgery~\cite{horsman2012SurfaceCode} in particular has attracted attention as a method that can efficiently implement multi-qubit logical operations. Roughly speaking, lattice surgery involves allocating a logical qubit on a vertex of a graph determined by a target architecture, with each logical operation depicted as an operation on the graph. It is established that lattice surgery with several simple operations can execute a universal gate set, thereby enabling universal fault-tolerant quantum computation.

An essential operation in lattice surgery is the \emph{merge} operation, facilitating logical operations on two or more logical qubits. The merge operation establishes a connection along a path between the target qubits positioned on the vertices. The point here is that no other logical qubit must be assigned to any vertex (except endpoints) included in the merge path. If no path meeting this condition exists during a merge operation, the execution of a quantum program will halt. Therefore, the compiler must schedule instructions to prevent situations where no merge paths exist during execution.

Quantum compilers typically involve multiple optimization passes, and the improper combination of these passes may lead to compiled programs that no longer adhere to connectivity constraints. Therefore, it is essential to statically verify that the compilation result indeed satisfies these constraints. Additionally, such verification tools should be capable of addressing quantum programs with high-level features, such as function calls and branches, as fault-tolerant quantum computation programs often tend to be large, making circuit-based methods impractical for scaling. However, to the best of our knowledge, there is currently no formal verification framework specifically tailored for lattice surgery.

To tackle this challenge, we present a type-based verification approach to ensuring the satisfaction of connectivity constraints between logical qubits in lattice surgery. Our contributions primarily encompass the following components: (1) the introduction of $\langname$, a first-order quantum programming language whose operational semantics reflects graph operations in lattice surgery, (2) a type system ensuring that well-typed $\langname$ programs inherently adhere to the connectivity constraints during execution, and (3) a type checking algorithm grounded in the offline dynamic connectivity problem. Although the details are not covered in this paper, we have also implemented our approach in Rust and applied it to several examples.\footnote{The source code is available at \url{https://github.com/SoftwareFoundationGroupAtKyotoU/tysurgery}.}

This paper is organized as follows: \cref{sec:background} explains the background of this work, including the basics of quantum computing and fault-tolerant quantum computation by lattice surgery. \cref{sec:motivation} gives a motivating example. \cref{sec:qtile} introduces $\langname$ to describe lattice surgery's operations and semantics. \cref{sec:typing} formalizes a type system for the verification of connectivity constraints. \cref{sec:extensions} discusses how to extend our language. \cref{sec:related-work} provides the related work, and finally, \cref{sec:conclusion} concludes this paper with future directions. \forconf{For full definitions and proofs, readers are referred to a full version available at \url{https://www.fos.kuis.kyoto-u.ac.jp/~wakizaka/aplas24-tysurgery.pdf}.}

\section{Background}\label{sec:background}

\subsection{The Basics of Quantum Computation}

In quantum computation, a \emph{qubit} is the unit of information.
A \emph{quantum state} of a qubit is a normalized vector of the 2-dimensional Hilbert space $\mathcal{H} \cong \mathbb{C}^2$.
Each state is represented by $\alpha\ket{0} + \beta\ket{1}$, where $\alpha, \beta \in \mathbb{C}$ satisfying $|\alpha|^2 + |\beta|^2 = 1$ and $\{\ket{0}, \ket{1}\}$ denotes the standard basis vectors of $\mathbb{C}^2$.
Here, we use \emph{Dirac notation}, which encloses integers or variables with $|$ and $\rangle$, to denote quantum states.
The state of $n$ qubits is a normalized vector of the tensor product $\bigotimes_{i=1}^n \mathbb{C}^2 \cong \mathbb{C}^{2^n}$.
For example, if $\ket{\psi} = \ket{0}$ and $\ket{\phi} = \frac{1}{\sqrt{2}}(\ket{0} + \ket{1})$,
then $\ket{\psi} \otimes \ket{\phi} = \ket{\psi}\ket{\phi} = \frac{1}{\sqrt{2}}(\ket{00} + \ket{01})$.
We call a quantum state $\ket{\psi} \in \mathbb{C}^{2^n}$ a \emph{pure state}.
On the other hand, when we have one of the quantum states $\{\ket{\psi_i}\}$ generated randomly with probabilities $p_i$ and do not know which state was generated, we call the state $\{(p_i, \ket{\psi_i})\}$ \emph{mixed state}.
A mixed state can be described by a \emph{density operator} $\rho = \sum_i p_i \ket{\psi_i}\bra{\psi_i}$, where $\bra{\psi_i}$ is the adjoint of $\ket{\psi_i}$ and $\sum_i p_i = 1$.
We often use density operators to describe quantum states because they can uniformly express both pure and mixed states.
We write $\mathcal{S}(\mathcal{H})$ for the set of density operators.

A quantum state can be manipulated by unitary operators called \emph{quantum gates}.
For example, the Hadamard gate and the $\GateCX$ (controlled $X$) gate are defined by $H\ket{x} = \frac{1}{\sqrt{2}}(\ket{0} + (-1)^x\ket{1})$, $\GateCX \ket{x}\ket{y} = \ket{x}\ket{x \oplus y}$, where $x, y \in \{0, 1\}$ and $\oplus$ denotes the Boolean XOR operation. The \emph{Pauli operators} defined by $\mathcal{P}_n = \{I, X, Y, Z\}^{\otimes n}$ for a $n$ qubits system are also important quantum operations. The operator on density operators corresponding to $U$ is described as a super operator $U[\cdot]U^\dagger : \mathcal{S}(\mathcal{H}) \rightarrow \mathcal{S}(\mathcal{H})$, where $U^\dagger$ is the adjoint of $U$.

An instruction set architecture generally provides a \emph{universal} gate set, a subset of quantum gates that realizes (approximate) universal quantum computation. For example, $\{H, T, \GateCX\}$ is a well-known universal quantum gate set.

To get the result of quantum computation, we have to perform \emph{quantum measurements} which consist of measurement operators $M_1, M_2, \dots, M_n$ acting on the state space and satisfying $\sum_i M_i^\dagger M_i = I$. When performing measurements to a quantum state $\rho$, we get an outcome corresponding to one of $M_i$ with probability $p_i = \Tr(M_i^\dagger M_i\rho)$, and then the quantum state is changed to $\frac{M_i\rho M_i^\dagger}{p_i}$. The measurements defined by $\mathcal{M}_P = \{(I+(-1)^sP)/2\}_{s=0,1}$ for a Pauli operator $P \in \mathcal{P}_n$ is called Pauli measurements, which plays an important role in fault-tolerant quantum computation.

\subsection{Fault-Tolerant Quantum Computation with Lattice Surgery}\label{sec:lattice-surgery}

This section explains fault-tolerant quantum computation employing lattice surgery. We utilize surface codes as an illustrative example, although analogous principles apply to other topological codes. We note that the error correction procedure is omitted in this paper as it is unnecessary for comprehending this study.

Quantum error-correcting codes serve to protect quantum data against quantum noises by constructing a logical qubit from noisy physical qubits. However, for these codes to be implemented on real quantum devices, they must satisfy various architectural constraints, particularly connectivity constraints. These constraints dictate that two-qubit gates, such as the $\GateCX$ gate, can only be applied to pairs of directly connected qubits. Therefore, the error correction process itself must also satisfy these locality conditions. For example, the left-hand side of \cref{fig:quantum-chip} represents an architecture in which physical qubits (white circles) are arranged in two dimensions, with only nearest-neighbor interactions allowed, and such configuration is standard in quantum computers~\cite{beverland2022AssessingRequirements,lee2021EvenMore,chamberland2021UniversalQuantum}.

The surface code~\cite{fowler2012SurfaceCodes} is a topological code that satisfies locality conditions and holds promise for future implementation. A logical qubit encoded by surface codes is represented by a rectangle of physical qubits, as depicted on the right-hand side of \cref{fig:quantum-chip}. In \cref{fig:quantum-chip}, there are two logical qubits, $q_0$ and $q_1$, where filled circles represent physical qubits used to construct a logical qubit. In this way, multiple logical qubits can be created using some of the regions of the physical qubits on a quantum computer. Logical qubits can be arranged freely as long as they do not overlap each other's proprietary areas, but in practice, it is customary to arrange them so that their relative positions are simple to facilitate hardware control and compiler optimization. Additionally, physical qubits represented by unfilled circles can be used as auxiliary qubits to implement logical operations on logical qubits, as explained later.

\begin{figure}[tb]
  \centering
  \includegraphics[width=12.0cm]{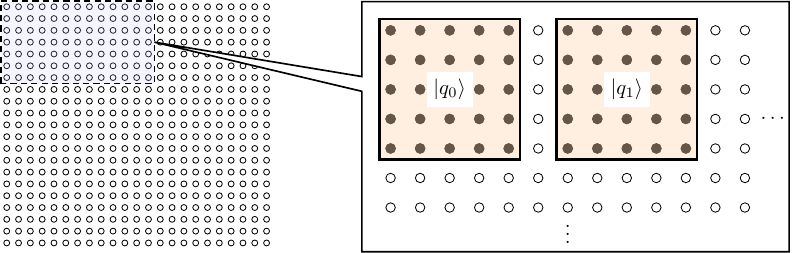}
  \caption{(Left) A 2D layout of physical qubits on a quantum computer. (Right) Two logical qubits constructed with the surface code.}
  \label{fig:quantum-chip}
\end{figure}

To achieve fault-tolerant quantum computation, logical operations on logical qubits must also be accomplished. In this paper, we briefly describe logical operations that can be performed with low latency and, when combined, enable universal quantum computation. These logical operations will later serve as the instruction set of $\langname$. For more details, readers are referred to literatures~\cite{fowler2019LowOverhead,litinski2019GameSurface}.

\paragraph{Qubit allocation/deallocation.}
The qubit initialization to the logical $\ket{0}$ state can be achieved by occupying an unused area of the physical qubits and initializing the state. Initialization can also be done to a state $\ket{m} = (\ket{0} + e^{i\pi/4}\ket{1})/\sqrt{2}$, known as the magic state, which is necessary for implementing the logical $T$ gate required for universal quantum computation\footnote{This process is called magic state injection and does not actually initialize the qubit in the exact $\ket{m}$ state. To generate $\ket{m}$ states with negligible approximation errors, it is necessary to use a protocol called magic state distillation~\cite{knill2004FaultTolerantPostselected,bravyi2005UniversalQuantum}.}.  In contrast, qubit deallocation is achieved by making the allocated area unused after performing the single qubit Pauli measurement described below.

\paragraph{Single qubit operations.}
The Pauli gates $X$, $Z$, the Hadamard gate $H$ and the phase gate $S$ are logical operations that can be easily realized on topological codes. Additionally, the logical Pauli measurements $M_Z$ and $M_X$ for a single logical qubit can be performed.

\paragraph{Multi-qubit operations.}
Multi-qubit operations can be realized by lattice surgery in a manner that satisfies the locality condition of physical qubits. Specifically, in lattice surgery, Pauli measurements $\mathcal{M}_P$ are achieved through operations known as \emph{merge} and \emph{split} operations. As shown in \cref{fig:merge-split}, the merge operation connects the target logical qubits using auxiliary physical qubits in between, while the split operation disconnects them via an appropriate physical measurement operation. It is important to note that sufficient free space between the target qubits is required to perform the merge operation. Such space is released by the split operation immediately after the merge operation. Multi-qubit Pauli measurements can be used, for example, to implement the logical $\GateCX$ gate (\cref{fig:impl-cx}) and the logical $T$ gate, which enable universal quantum computation.

\begin{remark}
Strictly speaking, which boundaries of the logical qubits (e.g., the four sides of each rectangle in \cref{fig:merge-split}) can be used for a merge operation depends on the basis of measurement. For simplicity, we will ignore this constraint in this paper. However, it is straightforward to extend our proposed method to account for this constraint.
\end{remark}

\begin{figure}[tb]
  \centering
  \includegraphics[scale=0.9]{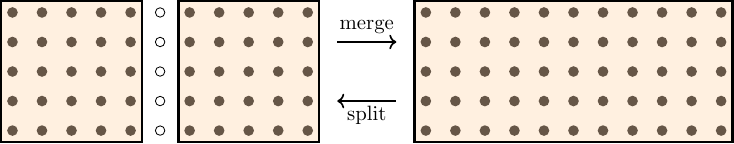}
  \caption{The merge and split operation for the surface code.}
  \label{fig:merge-split}
\end{figure}

\begin{figure}[tb]
  \centering
  \begin{minipage}{0.47\columnwidth}
    \centering
    \includegraphics[height=3.0cm]{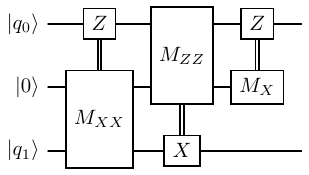}
    \caption{The implementation of the $\GateCX$ gate.}
    \label{fig:impl-cx}
  \end{minipage}
  \hspace{0.04\columnwidth}
  \begin{minipage}{0.47\columnwidth}
    \centering
    \includegraphics[height=3.0cm]{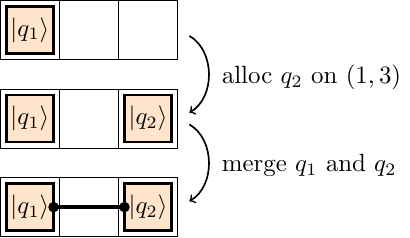}
    \caption{The abstracted surgery operations.}
    \label{fig:simplified-surgery}
  \end{minipage}
\end{figure}

To summarize this section, we provide an example of how surgery operations on surface codes proceed in \cref{fig:simplified-surgery}. In the case of surface codes, the physical qubits are divided into compartments on the grid, and logical qubits are allocated so that they just fit into their respective cells. This abstracts the merge operation to the process of connecting target cells on the graph using free paths between them. Note that we can merge two adjacent cells because there is a thin gap between them as shown in \cref{fig:quantum-chip}. The $\langname$ language is formalized using this abstracted execution model. From now on, abstracted graphs (e.g., grid graphs) will be referred to as architecture graphs.

\section{Motivating Example}\label{sec:motivation}

We illustrate instances where a quantum program utilizing surgery operations gets stuck by presenting examples depicted in \cref{fig:deadlock-surgery}. In these scenarios, quantum programs attempt to manipulate qubits, represented as squares on a 2D grid graph. For example, in the first scenario (\cref{fig:example-merge-ok}), the program tries to merge $q_1$ and $q_4$. This operation succeeds as an accessible path between $q_1$ and $q_4$ exists, as indicated by the black line.

Conversely, in the remaining examples, their merge operations fail. The second example (\cref{fig:example-merge-ng1}) attempts to merge $q_1$ and $q_4$ but fails because qubit $q_2$ interrupts their connectivity, unlike the first scenario. The third example (\cref{fig:example-merge-ng2}) is more intricate than the others. In \cref{fig:example-merge-ng2}, $\GateCX(q_1, q_4)$ and $\GateCX(q_2,q_3)$ are depicted midway through execution. In reality, a situation akin to \cref{fig:example-merge-ng2} unfolds through the following steps: (1) we express $\GateCX(q_1,q_4)$ and $\GateCX(q_2,q_3)$ in a source language in this sequence, (2) a compiler transpiles them into surgery operations, decomposing the $\GateCX$ gate (\cref{fig:impl-cx}), (3) and a transpiler pass erroneously reorders the instruction allocating an ancilla $\ket{0}$ for $\GateCX(q_2, q_3)$ before completing $\GateCX(q_1, q_4)$. Ultimately, in both examples, the programs fail to execute the merge operation, leading to abnormal termination.

\begin{figure}[tb]
  \begin{tabular}{ccc}
    \begin{minipage}{0.33\hsize}
      \centering
      \begin{tikzpicture}
        \draw (0, 0) grid (3, 2);
        \logiQSq{0.1}{1.1}{0.8}{orange!20}{$\ket{q_1}$}
        \logiQSq{0.1}{0.1}{0.8}{orange!20}{$\ket{q_2}$}
        \logiQSq{2.1}{1.1}{0.8}{orange!20}{$\ket{q_3}$}
        \logiQSq{2.1}{0.1}{0.8}{orange!20}{$\ket{q_4}$}
        \draw[ultra thick] (0.9,1.5)--(1.5,1.5)--(1.5,0.5)--(2.1,0.5);
        \fill (0.9,1.5) circle [radius=0.09];
        \fill (2.1,0.5) circle [radius=0.09];
      \end{tikzpicture}
      \subcaption{}
      \label{fig:example-merge-ok}
    \end{minipage}
    \begin{minipage}{0.33\hsize}
      \centering
      \begin{tikzpicture}
        \draw (0, 0) grid (3, 2);
        \logiQSq{0.1}{1.1}{0.8}{orange!20}{$\ket{q_1}$}
        \logiQSq{2.1}{1.1}{0.8}{orange!20}{$\ket{q_3}$}
        \logiQSq{2.1}{0.1}{0.8}{orange!20}{$\ket{q_4}$}
        \logiQSq{1.1}{0.1}{0.8}{orange!20}{$\ket{q_2}$}
        \draw[ultra thick, dotted] (0.9,1.5)--(1.5,1.5)--(1.5,0.5)--(2.1,0.5);
        \fill (0.9,1.5) circle [radius=0.09];
        \fill (2.1,0.5) circle [radius=0.09];
      \end{tikzpicture}
      \subcaption{}
      \label{fig:example-merge-ng1}
    \end{minipage}
    \begin{minipage}{0.33\hsize}
      \centering
      \begin{tikzpicture}
        \draw (0, 0) grid (3, 2);
        \logiQSq{0.1}{1.1}{0.8}{orange!20}{$\ket{q_1}$}
        \logiQSq{0.1}{0.1}{0.8}{orange!20}{$\ket{q_2}$}
        \logiQSq{2.1}{1.1}{0.8}{orange!20}{$\ket{q_3}$}
        \logiQSq{2.1}{0.1}{0.8}{orange!20}{$\ket{q_4}$}
        \logiQSq{1.1}{1.1}{0.8}{orange!20}{$\ket{0}$}
        \logiQSq{1.1}{0.1}{0.8}{orange!20}{$\ket{0}$}

        \draw[ultra thick, dotted] (1.76, 0.9) -- (1.76, 1.5) -- (2.1, 1.5);
        \fill (1.76, 0.9) circle [radius=0.09];
        \fill (2.1, 1.5) circle [radius=0.09];
        \draw[ultra thick, dotted] (1.24, 1.1) -- (1.24, 0.26) -- (2.1, 0.26);
        \fill (1.24, 1.1) circle [radius=0.09];
        \fill (2.1, 0.26) circle [radius=0.09];
      \end{tikzpicture}
      \subcaption{}
      \label{fig:example-merge-ng2}
    \end{minipage}
  \end{tabular}
  \caption{Making attempts to merge qubits in several situations.}
  \label{fig:deadlock-surgery}
\end{figure}
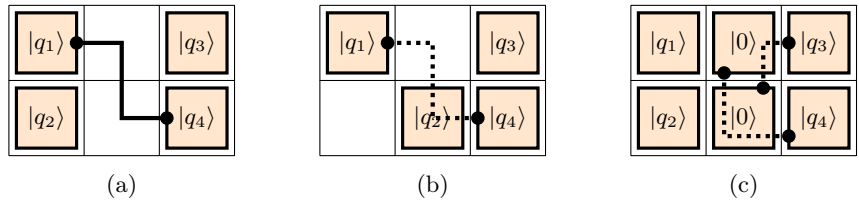

As discussed above, the execution of a quantum program utilizing lattice surgery may halt for various reasons, notably inappropriate compiler strategies in qubit allocation and merge path scheduling. Regrettably, such issues are likely to persist because the availability of nodes for logical qubits and merge paths will remain limited even in the future. Hence, we advocate for the necessity of a framework to statically verify whether a given program can complete its execution without getting stuck at any point.

\section{$\langname$ : Quantum PL for Lattice Surgery}\label{sec:qtile}

Before presenting our verification methodology, we introduce $\langname$, an imperative quantum programming language encompassing lattice surgery operations, mutable references, and first-order functions. $\langname$ is primarily designed to serve as a target language for quantum compilers.

\subsection{Syntax}\label{sec:language}

\begin{figure}[tb]
  \begin{align*}
    \text{Location variables} & & l \in &\ \mathit{Locations} \\
    \text{Function declarations} & & d \coloneqq &\  \langmv{f}  \mapsto [   \overline{ \langmv{l} }   ]( \langmv{x_{{\mathrm{1}}}}  \langsym{,} \, .. \, \langsym{,}  \langmv{x_{\langmv{n}}} ) \langnt{e}  \\
    \text{Expressions} & & \langnt{e} \coloneqq &\ \langmv{x} \mid \langkw{let} \, \langmv{x}  \langsym{=}  \langkw{init} \, \langsym{(}  \langmv{l}  \langsym{)} \, \langkw{in} \, \langnt{e} \mid \langkw{let} \, \langmv{x}  \langsym{=}  \langkw{minit} \, \langsym{(}  \langmv{l}  \langsym{)} \, \langkw{in} \, \langnt{e} \mid \\
                       & & &\ \langkw{free} \, \langmv{x}  \langsym{;}  \langnt{e} \mid \langkw{let} \, \langmv{y}  \langsym{=}   M_{ B_{{\mathrm{1}}}  \langsym{,} \, .. \, \langsym{,}  B_{\langmv{n}} }   \langsym{(}  \langmv{x_{{\mathrm{1}}}}  \langsym{,} \, .. \, \langsym{,}  \langmv{x_{\langmv{n}}}  \langsym{)} \, \langkw{in} \, \langnt{e} \mid \langmv{U}  \langsym{(}  \langmv{x}  \langsym{)} \mid \\
                         & & &\ \langkw{let} \, \langmv{x}  \langsym{=}  \langkw{mkref} \, \langnt{e} \, \langkw{in} \, \langnt{e} \mid \langsym{*}  \langnt{e} \mid \langmv{x}  \langsym{:=}  \langnt{e} \mid \langnt{e_{{\mathrm{1}}}}  \langsym{;}  \langnt{e_{{\mathrm{2}}}} \mid \\
      & & &\ \langkw{if} \, \langnt{e_{{\mathrm{1}}}} \, \langkw{then} \, \langnt{e_{{\mathrm{2}}}} \, \langkw{else} \, \langnt{e_{{\mathrm{3}}}} \mid \langkw{while} \, \langnt{e_{{\mathrm{1}}}} \, \langkw{do} \, \langnt{e_{{\mathrm{2}}}} \mid \langmv{f}  \langsym{[}   \overline{ \langmv{l} }   \langsym{]}  \langsym{(}  \langmv{x_{{\mathrm{1}}}}  \langsym{,} \, .. \, \langsym{,}  \langmv{x_{\langmv{n}}}  \langsym{)} \\
    \text{Values} & & v \coloneqq &\ \langmv{x} \mid \langmv{l} \mid \langkw{unit} \mid  \texttt{true}  \mid  \texttt{false}  \\
    \text{Program} & & P \coloneqq &\  \braket{  \langsym{\{}  \langmv{d_{{\mathrm{1}}}}  \langsym{,} \, .. \, \langsym{,}  \langmv{d_{\langmv{n}}}  \langsym{\}} ,  \langnt{e}  }  \\
    \text{Measurement basis} & & B \coloneqq &\  X  \mid  Z  \\
    \text{Gates} & & \langmv{U} \coloneqq &\  X  \mid  Z  \mid \langmv{H} \mid  S 
  \end{align*}
  \caption{The grammar of $\langname$.}
  \label{fig:syntax}
\end{figure}

The grammar of $\langname$ is given in \cref{fig:syntax}.
We assume that an architecture configuration is represented as a graph $G = (V, E)$, where $V \subseteq \mathit{Locations}$, and a location variable, denoted by $\langmv{l}$, signifies a location where qubits can be allocated. We use $\langmv{l} \in \mathit{Locations} \setminus V$ to denote a function parameter that is substituted with a specific location upon function invocation. We write $ \overline{ \langmv{l} } $ for a sequence of location variables $\langmv{l_{{\mathrm{1}}}}  \langsym{,} \, .. \, \langsym{,}  \langmv{l_{\langmv{n}}}$.
We also use $\langmv{L}$ to denote a set of location variables.


There are four primitives for quantum operations in $\langname$: qubit allocation, deallocation, unitary operations, and quantum measurements. Qubit allocating operations $\langkw{init} \, \langsym{(}  \langmv{l}  \langsym{)}$ and $\langkw{minit} \, \langsym{(}  \langmv{l}  \langsym{)}$ initialize a location $\langmv{l}$ with a logical state $\ket{0}$ and a magic state $\ket{m}$, respectively. In $\langname$, all locations of qubits in a program are statically determined at compile time. A unitary operation $\langmv{U}  \langsym{(}  \langmv{x}  \langsym{)}$ applies a unitary gate $U$ to $\langmv{x}$, where $U$ must be one of $H, X, Z$, and $S$ gate. Note that the $T$ gate for universal quantum computation can be realized through a combination of if expressions, basic gates, and measurements with a magic state.

Our language supports quantum measurements with one or two qubits. We can measure qubits with different bases per each qubit; for example, $ M_{  X   \langsym{,}   Z  }   \langsym{(}  \langmv{x_{{\mathrm{1}}}}  \langsym{,}  \langmv{x_{{\mathrm{2}}}}  \langsym{)}$ measures $\langmv{x_{{\mathrm{1}}}}  \langsym{,}  \langmv{x_{{\mathrm{2}}}}$ with $X, Z$ bases, respectively. Although a two-qubit measurement is implemented through merge and split operations, as explained in \cref{sec:lattice-surgery}, our language does not explicitly specify a path for a merge operation between the target qubits. Instead, we leave the pathfinding problem to the runtime. It is straightforward to extend the language with measurements specifying paths.

A program is represented as a pair $ \braket{  D ,  \langnt{e}  } $, where $D = \{\langmv{d_{{\mathrm{1}}}}  \langsym{,} \, ... \, \langsym{,}  \langmv{d_{\langmv{n}}}\}$ constitutes a set of first-order and non-recursive function definitions, and $\langnt{e}$ denotes the program entry point. A function declaration $\langmv{d} =  \langmv{f}  \mapsto [   \overline{ \langmv{l} }   ]( \langmv{x_{{\mathrm{1}}}}  \langsym{,} \, .. \, \langsym{,}  \langmv{x_{\langmv{n}}} ) \langnt{e} $ maps a function name $f$ to a tuple of location variables $ \overline{ \langmv{l} } $ and argument names $\langmv{x_{{\mathrm{1}}}}  \langsym{,} \, ... \, \langsym{,}  \langmv{x_{\langmv{n}}}$ bound within the function body $\langnt{e}$. The parameters $ \overline{ \langmv{l} } $ in a function declaration allow us to call a function in various scenarios where the topology of qubits used within its body varies.

\begin{wrapfigure}[14]{l}{0.48\textwidth}
  \vspace{-24pt}
  \begin{lstlisting}[numbers=left,numbersep=3pt,numberstyle=\tiny\color{black},numberblanklines=false]
[l0,l1,l2]
cx(q0:qbit(l0), q1:qbit(l1)) {
  let aux = init(l2) in$\label{line:internal-init}$
  let a = meas[X,X](aux, q1) in
  (if a then Z(q0));
  let b = meas[Z,Z](q0, aux) in
  (if b then X(q1));
  let c = meas[X](aux) in
  (if c then Z(q0));
  free(aux)$\label{line:free-aux}$
}
  \end{lstlisting}
  \caption{An implementation of $\GateCX$}
  \label{fig:example-of-cx}
\end{wrapfigure}
For example, an implementation of the $\GateCX$ gate in $\langname$ is provided in \cref{fig:example-of-cx}. The function {\tt cx} features location parameters $\langmv{l_{{\mathrm{0}}}}  \langsym{,}  \langmv{l_{{\mathrm{1}}}}  \langsym{,}  \langmv{l_{{\mathrm{2}}}}$ and accepts two qubits as arguments. The location variable $\langmv{l_{{\mathrm{2}}}}$ is employed for internal qubit allocation (see line \ref{line:internal-init}). The ancilla qubit {\tt aux} is deallocated at the end of {\tt cx} (see line \ref{line:free-aux}), allowing for the reuse of $\langmv{l_{{\mathrm{2}}}}$ after the function call. Notably, we use an if-expression without an else clause, which is a syntax sugar representing an if-expression where the else clause is a unit expression.

\subsection{Semantics}\label{sec:semantics}

We design operational semantics of our language to model the execution of quantum programs with lattice surgery.
The complete definition is given in \cref{fig:semantics}.
We define the operational semantics as a transition relation on \emph{runtime states} denoted by a triple \rstate{ H }{\rho}{\langnt{e}};
$H$ is a partial function from $\mathit{Variables}$ to $\mathit{Values}$, $\rho$ is a quantum state, and $\langnt{e}$ is a reducing expression.
Then we write the transition relation as $\evalarrow$, where $D$ contains function definitions and $G$ is an architecture graph.
From here, we explain briefly each transition rule.

A qubit allocation $\langkw{let} \, \langmv{x}  \langsym{=}  \langkw{init} \, \langsym{(}  \langmv{l}  \langsym{)} \, \langkw{in} \, \langnt{e}$ allocates a qubit at $\langmv{l}$ and binds a qubit variable $\langmv{x}$ to $\langmv{l}$. The semantics of $\langkw{minit} \, \langsym{(}  \langmv{l}  \langsym{)}$ is defined similarly. Both expressions require that the location $\langmv{l}$ is empty. On the other hand, an expression $\langkw{free} \, \langmv{x}$ deallocates a qubit $\langmv{x}$ on a location $\langmv{l} = H(x)$ and removes $\langmv{x}$ from a current heap $H$. The quantum state after deallocating a qubit at $\langmv{l}$ is the partial trace over the qubit corresponding to $\langmv{l}$. For example, deallocating a qubit at $\langmv{l_{{\mathrm{1}}}}$, which is one of the Bell pair $\lblketbra{\Phi_+}{\Phi_+}{\langmv{l_{{\mathrm{1}}}}  \langsym{,}  \langmv{l_{{\mathrm{2}}}}}$ where $\ket{\Phi_+} = (\ket{00} + \ket{11})/\sqrt{2}$, produces a mixed state $\frac{1}{2}\lblketbra{0}{0}{\langmv{l_{{\mathrm{1}}}}} + \frac{1}{2}\lblketbra{1}{1}{\langmv{l_{{\mathrm{1}}}}}$.

An expression $\langkw{let} \, \langmv{x}  \langsym{=}   M_{ \langmv{B_{{\mathrm{1}}}}  \langsym{,} \, .. \, \langsym{,}  \langmv{B_{\langmv{n}}} }   \langsym{(}  \langmv{x_{{\mathrm{1}}}}  \langsym{,} \, .. \, \langsym{,}  \langmv{x_{\langmv{n}}}  \langsym{)} \, \langkw{in} \, \langnt{e}$ performs quantum measurements to $\langmv{x_{{\mathrm{1}}}}  \langsym{,} \, .. \, \langsym{,}  \langmv{x_{\langmv{n}}}$ with basis $\langmv{B_{{\mathrm{1}}}}  \langsym{,} \, .. \, \langsym{,}  \langmv{B_{\langmv{n}}}$ ($n = 1$ or $2$). The variable $\langmv{x}$ is bound to $\langmv{v} \in \{0, 1\}$, the result of the measurement, and the quantum state $\rho$ changes to $M_v\rho M_v^\dagger/p_v$. The point is that two qubit measurements require that there exists a free path between the target locations $\langmv{l_{{\mathrm{1}}}}$ and $\langmv{l_{{\mathrm{2}}}}$. This requirement is formalized as $G  \mid  \langmv{L}  \vDash  \langmv{l_{{\mathrm{1}}}}  \sim  \langmv{l_{{\mathrm{2}}}}$, indicating that there is a path between $\langmv{l_{{\mathrm{1}}}}$ and $\langmv{l_{{\mathrm{2}}}}$ on a graph $G$, where $\langmv{L}$ represents a set of locations where any qubits are allocated. If such a path does not exist, the program gets stuck. On the other hand, single-qubit measurements always succeed.

We omit other expressions in this paper, as their semantics are similar to those in common programming languages.

\begin{figure}[tb]
  \fbox{$  \left[  \langmv{H} ,  \rho ,  \langnt{e}  \right]   \rightarrow_{ \langmv{D} ,  G }   \left[  \langmv{H'} ,  \rho' ,  \langnt{e'}  \right]  $}

  \infrule[E-Init]{
    x' \not\in \dom(H)
    \andalso \langmv{l} \not\in \cod(H)
    \andalso \langmv{l} \in G
  }{
    \rstate{ H }{ \rho }{ \langkw{let} \, \langmv{x}  \langsym{=}  \langkw{init} \, \langsym{(}  \langmv{l}  \langsym{)} \, \langkw{in} \, \langnt{e} }
    \evalarrow \rstate{ H \{ x' \mapsto \langmv{l} \} }{ \rho \otimes \lblketbra{0}{0}{\langmv{l}} }{ \langsym{[}  \langmv{x'}  \slash  \langmv{x}  \langsym{]} \, \langnt{e} }
  }

  \infrule[E-MInit]{
    x' \not\in \dom(H)
    \andalso \langmv{l} \not\in \cod(H)
    \andalso \langmv{l} \in G
  }{
    \rstate{ H }{ \rho }{ \langkw{let} \, \langmv{x}  \langsym{=}  \langkw{minit} \, \langsym{(}  \langmv{l}  \langsym{)} \, \langkw{in} \, \langnt{e} }
    \evalarrow \rstate{ H \{ x' \mapsto \langmv{l} \} }{ \rho \otimes \lblketbra{T}{T}{\langmv{l}} }{ \langsym{[}  \langmv{x'}  \slash  \langmv{x}  \langsym{]} \, \langnt{e} }
  }

  \begin{tabular}{c}
    \begin{minipage}{0.49\columnwidth}
      \infrule[E-Free]{
        H(\langmv{x}) = \langmv{l}
      }{
        \rstate{ H }{ \rho }{ \langkw{free} \, \langmv{x}  \langsym{;}  \langnt{e} }
        \evalarrow \rstate{ H \setminus \langmv{x} }{ \Tr_{\langmv{l}}(\rho) }{ \langnt{e} }
      }
    \end{minipage}
    \begin{minipage}{0.49\columnwidth}
      \infrule[E-Gate]{
        H(\langmv{x}) = \langmv{l}
      }{
        \rstate { H }{ \rho }{ \langmv{U}  \langsym{(}  \langmv{x}  \langsym{)} }
        \evalarrow
        \rstate { H }{ U_l\rho U_l^\dagger }{ () }
      }
    \end{minipage}
  \end{tabular}

  \infrule[E-Meas2]{
    H(\langmv{x_{{\mathrm{1}}}}) = \langmv{l_{{\mathrm{1}}}}
    \andalso H(\langmv{x_{{\mathrm{2}}}}) = \langmv{l_{{\mathrm{2}}}}
    \andalso M_s \coloneqq (I + (-1)^s B_1 B_2)/2 \ (s = 0, 1) \\
    \langmv{x'} \not\in \dom(H)
    \andalso \langmv{v} \in \{0, 1\}
    \andalso p_v \coloneqq \Tr\left(M_v^\dagger M_v \rho\right) \neq 0
    \andalso G  \mid   \langnt{V}  \langsym{(}  G  \langsym{)}  \setminus   \mathrm{Used}( \langmv{H} )    \vDash  \langmv{l_{{\mathrm{1}}}}  \sim  \langmv{l_{{\mathrm{2}}}} 
  }{
    \rstate{ H }{ \rho }{ \langkw{let} \, \langmv{x}  \langsym{=}   M_{ B_{{\mathrm{1}}}  \langsym{,}  B_{{\mathrm{2}}} }   \langsym{(}  \langmv{x_{{\mathrm{1}}}}  \langsym{,}  \langmv{x_{{\mathrm{2}}}}  \langsym{)} \, \langkw{in} \, \langnt{e} }
    \evalarrow
    \rstate{ H\{x' \mapsto v\} }{ M_v \rho M_v^\dagger / p_v }{ \langsym{[}  \langmv{x'}  \slash  \langmv{x}  \langsym{]} \, \langnt{e} }
  }

  \infrule[E-Call]{
     \langmv{f}  \mapsto [   \overline{ \langmv{l'} }   ]( \langmv{x'_{{\mathrm{1}}}}  \langsym{,} \, .. \, \langsym{,}  \langmv{x'_{\langmv{n}}} ) \langnt{e'}  \in D
  }{
    \rstate{ H }{ \rho }{ \langmv{f}  \langsym{[}   \overline{ \langmv{l} }   \langsym{]}  \langsym{(}  \langmv{x_{{\mathrm{1}}}}  \langsym{,} \, .. \, \langsym{,}  \langmv{x_{\langmv{n}}}  \langsym{)} }
    \evalarrow 
    \rstate{ H }{ \rho }{ \langsym{[}   \overline{ \langmv{l} }   \slash   \overline{ \langmv{l'} }   \langsym{]} \, \langsym{[}  \langmv{x_{{\mathrm{1}}}}  \slash  \langmv{x'_{{\mathrm{1}}}}  \langsym{]} \, .. \, \langsym{[}  \langmv{x_{\langmv{n}}}  \slash  \langmv{x'_{\langmv{n}}}  \langsym{]} \, \langnt{e'} }
  }

  \caption{Operational semantics (excerpt).}
  \label{fig:semantics}
\end{figure}

\section{Typing}\label{sec:typing}

This section presents the definition of a type system for $\langname$. The aim of our type system is to ensure that well-typed programs do not halt execution due to surgery operations. We proved this property as type soundness.

\subsection{Types, Environments and Commands}

\begin{figure}[tb]
  \begin{align*}
    \text{Types} & & \tau \coloneqq &\  \texttt{qbit}( l )  \mid \langkw{unit} \mid  \texttt{bool}  \mid \langkw{ref} \, \tau \mid \dots \\
    \text{Function Types} & & \theta \coloneqq &\  \Pi   \overline{ \langmv{l} }   . \braket{ \langmv{x_{{\mathrm{1}}}}  \langsym{:}  \tau_{{\mathrm{1}}}  \langsym{,} \, .. \, \langsym{,}  \langmv{x_{\langmv{n}}}  \langsym{:}  \tau_{\langmv{n}} } \xrightarrow{  \overline{C}  } \braket{ \Gamma  |  \tau }  \\
    \text{Typing Environment} & & \Gamma \coloneqq &\  \bullet  \mid \Gamma  \langsym{,}  \langmv{x}  \langsym{:}  \tau \mid \Gamma  \langsym{,}  \langmv{l}\\
    \text{Function Type Environment} & & \Theta \coloneqq &\  \bullet  \mid \Theta  \langsym{,}  \langmv{f}  \langsym{:}  \theta \\
    \text{Commands} & & C \coloneqq &\  \langmv{l_{{\mathrm{1}}}}  \sim  \langmv{l_{{\mathrm{2}}}}  \mid \langkw{alloc} \, \langsym{(}  \langmv{l}  \langsym{)} \mid \langkw{free} \, \langsym{(}  \langmv{l}  \langsym{)} \mid   \overline{ C_{{\mathrm{1}}} }  ^*  \mid   \overline{ C_{{\mathrm{1}}} }   \lor   \overline{ C_{{\mathrm{2}}} }  
  \end{align*}
  \caption{Types, environments and commands}
  \label{fig:type-syntax}
\end{figure}

The syntax of types, environments, and commands is provided in \cref{fig:type-syntax}. The type $ \texttt{qbit}( \langmv{l} ) $ denotes the type of a qubit, where $\langmv{l}$ represents the location of the qubit. Our type system distinguishes each qubit from others to analyze how qubits are manipulated within a given graph. The qubit type can be considered a kind of singleton type~\cite{hayashi1991SingletonUnion}. The other types are consistent with those commonly used in many programming languages: $\langkw{unit}$, $ \texttt{bool} $ and $\langkw{ref} \, \tau$ denote the type of the unit value, Boolean values, and references to a value of type $\tau$, respectively.

A type environment $\Gamma$ maps variables to types and also manages a set of qubit locations. The symbol $ \bullet $ denotes the empty environment. We use $\Gamma  \langsym{,}  \langmv{x}  \langsym{:}  \tau$ to denote the extension of environment $\Gamma$ with a type binding $\langmv{x}  \langsym{:}  \tau$, and $\Gamma  \langsym{,}  \langmv{l}$ to represent the extension of $\Gamma$ with a location $\langmv{l}$. If $\langmv{l} \in \Gamma$, it means that no qubit is allocated at $\langmv{l}$. We assume that all variable names and locations in $\Gamma$ are distinct. A function type environment $\Theta$ maps function names $\langmv{f}$ to function types $\theta$.

A command $C$ denotes an operation on locations, which can be perceived as a form of computational effects. We write $\overline{C}$ for a sequence of commands and $ \epsilon $ for the empty command sequence, indicating that an expression does not manipulate any locations. The concatenation of command sequences $\overline{C}_{{\mathrm{1}}}$ and $\overline{C}_{{\mathrm{2}}}$ is denoted by $\overline{C}_{{\mathrm{1}}}  \mathbin{+\mkern-10mu+}  \overline{C}_{{\mathrm{2}}}$. Commands $\langkw{alloc} \, \langsym{(}  \langmv{l}  \langsym{)}$ and $\langkw{free} \, \langsym{(}  \langmv{l}  \langsym{)}$ signify qubit allocation and deallocation, respectively, which are generated by expressions $\langkw{init} \, \langsym{(}  \langmv{l}  \langsym{)}$ and $\langkw{free} \, \langsym{(}  \langmv{l}  \langsym{)}$. A merge command $ \langmv{l_{{\mathrm{1}}}}  \sim  \langmv{l_{{\mathrm{2}}}} $ signifies an occurrence of a merge operation between $\langmv{l_{{\mathrm{1}}}}$ and $\langmv{l_{{\mathrm{2}}}}$, which is generated by measurements. Additionally, our language supports loops and branches, thus commands have constructors for them. A command $  \overline{ C }  ^* $ represents zero or more repetitions of $ \overline{ C } $, and $   \overline{ C_{{\mathrm{1}}} }    \lor   \overline{ C_{{\mathrm{2}}} }  $ denotes a branch that chooses either $ \overline{ C_{{\mathrm{1}}} } $ or $ \overline{ C_{{\mathrm{2}}} } $ at runtime.

We denote function types as $ \Pi   \overline{ \langmv{l} }   . \braket{ \langmv{x_{{\mathrm{1}}}}  \langsym{:}  \tau_{{\mathrm{1}}}  \langsym{,} \, .. \, \langsym{,}  \langmv{x_{\langmv{n}}}  \langsym{:}  \tau_{\langmv{n}} } \xrightarrow{  \overline{C}  } \braket{ \Gamma  |  \tau } $, where $ \overline{ \langmv{l} } $ represents a sequence of location parameters, $\langmv{x_{{\mathrm{1}}}}  \langsym{:}  \tau_{{\mathrm{1}}}  \langsym{,} \, .. \, \langsym{,}  \langmv{x_{\langmv{n}}}  \langsym{:}  \tau_{\langmv{n}}$ denote the types of arguments, $\Gamma$ represents a new typing environment, $\tau$ signifies the type of a return value, and $\overline{C}$ is a command sequence generated upon function invocation. Here, the environment $\Gamma$ is utilized to monitor which locations and variables remain valid after the function call. The parameter $ \overline{ \langmv{l} } $ enables function invocation in diverse contexts wherein the locations of qubits passed as arguments vary.

\subsection{Type System}

The typing rules are provided in \cref{fig:typing-quantum,fig:typing-classical}. A typing judgment has the form $\Theta  \mid  \Gamma  \vdash  \langnt{e}  \langsym{:}  \tau  \Rightarrow  \Gamma'  \mid  \overline{C}$, indicating that $\langnt{e}$ is well-typed under a function type environment $\Theta$ and typing environment $\Gamma$, evaluates to a value of type $\tau$, leading to a change in the type environment to $\Gamma'$, and generating a command sequence $\overline{C}$. Throughout the description of our type system, we use $ \mathrm{flv}( \Gamma ) $ and $ \mathrm{flv}( \tau ) $ to denote all free location variables in $\Gamma$ and $\tau$, respectively.

The typing rules for quantum expressions, which may generate several commands, are provided in \cref{fig:typing-quantum}. In \rn{T-Init} and \rn{T-MInit} for qubit allocation, a location $\langmv{l} \in \Gamma$ removed to create a type $ \texttt{qbit}( \langmv{l} ) $. Conversely, the rule \rn{T-Free} removes a qubit variable $\langmv{x}$ from the typing environment and returns the location variable $\langmv{l}$ associated with the type of $\langmv{x}$, allowing it to be reused for another qubit allocation. These two rules ensure that at most only one variable can have ownership over a location, and in this sense they are related to linear types~\cite{turner1995OnceType}. The rule \rn{T-Meas2} requires that the arguments are qubits and adds a command $ \langmv{l_{{\mathrm{1}}}}  \sim  \langmv{l_{{\mathrm{2}}}} $. It is noteworthy that it does not validate whether $\langmv{l_{{\mathrm{1}}}}$ and $\langmv{l_{{\mathrm{2}}}}$ are connected at this time. The \rn{T-Gate} and \rn{T-Meas1} rules do not generate any commands because single qubit measurements and applications of basic gates do not entail changes to qubit locations or connections between qubits.

\newlength{\tyvspace}
\setlength{\tyvspace}{6pt}

\begin{figure}[tb]
  \infrule[T-Init]{
    \Theta  \mid  \Gamma  \langsym{,}  \langmv{x}  \langsym{:}   \texttt{qbit}( \langmv{l} )   \vdash  \langnt{e}  \langsym{:}  \tau  \Rightarrow  \Gamma'  \mid  \overline{C}
  }{
    \Theta  \mid  \Gamma  \langsym{,}  \langmv{l}  \vdash  \langkw{let} \, \langmv{x}  \langsym{=}  \langkw{init} \, \langsym{(}  \langmv{l}  \langsym{)} \, \langkw{in} \, \langnt{e}  \langsym{:}  \tau  \Rightarrow  \Gamma'  \mid  \langkw{alloc} \, \langsym{(}  \langmv{l}  \langsym{)}  \mathbin{+\mkern-10mu+}  \overline{C}
  }

  \vspace{\tyvspace}
  \infrule[T-Free]{
    \Theta  \mid  \Gamma  \langsym{,}  \langmv{l}  \vdash  \langnt{e}  \langsym{:}  \tau  \Rightarrow  \Gamma'  \mid  \overline{C}
  }{
    \Theta  \mid  \Gamma  \langsym{,}  \langmv{x}  \langsym{:}   \texttt{qbit}( \langmv{l} )   \vdash  \langkw{free} \, \langmv{x}  \langsym{;}  \langnt{e}  \langsym{:}  \tau  \Rightarrow  \Gamma'  \mid  \langkw{free} \, \langsym{(}  \langmv{l}  \langsym{)}  \mathbin{+\mkern-10mu+}  \overline{C}
  }

  \vspace{\tyvspace}
  \infrule[T-Meas1]{
    \langmv{y}  \langsym{:}   \texttt{qbit}( l )  \in \Gamma
    \andalso \Theta  \mid  \Gamma  \langsym{,}  \langmv{x}  \langsym{:}   \texttt{bool}   \vdash  \langnt{e}  \langsym{:}  \tau  \Rightarrow  \Gamma'  \mid  \overline{C}
  }{
    \Theta  \mid  \Gamma  \vdash  \langkw{let} \, \langmv{x}  \langsym{=}   M_{ B }   \langsym{(}  \langmv{y}  \langsym{)} \, \langkw{in} \, \langnt{e}  \langsym{:}  \tau  \Rightarrow  \Gamma'  \mid  \overline{C}
  }

  \vspace{\tyvspace}
  \infrule[T-Meas2]{
    \langmv{x_{{\mathrm{1}}}}  \langsym{:}   \texttt{qbit}( l_{{\mathrm{1}}} )   \langsym{,}  \langmv{x_{{\mathrm{2}}}}  \langsym{:}   \texttt{qbit}( l_{{\mathrm{2}}} )  \in \Gamma
    \andalso \Theta  \mid  \Gamma  \langsym{,}  \langmv{x}  \langsym{:}   \texttt{bool}   \vdash  \langnt{e}  \langsym{:}  \tau  \Rightarrow  \Gamma'  \mid  \overline{C}
  }{
    \Theta  \mid  \Gamma  \vdash  \langkw{let} \, \langmv{x}  \langsym{=}   M_{ B_{{\mathrm{1}}}  \langsym{,}  B_{{\mathrm{2}}} }   \langsym{(}  \langmv{x_{{\mathrm{1}}}}  \langsym{,}  \langmv{x_{{\mathrm{2}}}}  \langsym{)} \, \langkw{in} \, \langnt{e}  \langsym{:}  \tau  \Rightarrow  \Gamma'  \mid   \langmv{l_{{\mathrm{1}}}}  \sim  \langmv{l_{{\mathrm{2}}}}   \mathbin{+\mkern-10mu+}  \overline{C}
  }

  \vspace{\tyvspace}
  \infrule[T-Gate]{
    \langmv{x}  \langsym{:}   \texttt{qbit}( \langmv{l} )  \in \Gamma
  }{
    \Theta  \mid  \Gamma  \vdash  \langmv{U}  \langsym{(}  \langmv{x}  \langsym{)}  \langsym{:}  \langkw{unit}  \Rightarrow  \Gamma  \mid   \epsilon 
  }

  \caption{Typing rules for quantum expressions.}
  \label{fig:typing-quantum}
\end{figure}

\begin{figure}[tb]
  \infrule[T-Var]{
    \langmv{x}  \langsym{:}  \tau \in \Gamma
  }{
    \Theta  \mid  \Gamma  \vdash  \langmv{x}  \langsym{:}  \tau  \Rightarrow  \Gamma  \mid   \epsilon 
  }

  \vspace{\tyvspace}
  \infrule[T-Seq]{
    \Theta  \mid  \Gamma  \vdash  \langnt{e_{{\mathrm{1}}}}  \langsym{:}  \langkw{unit}  \Rightarrow  \Gamma_{{\mathrm{1}}}  \mid   \overline{ C_{{\mathrm{1}}} } 
    \andalso \Theta  \mid  \Gamma_{{\mathrm{1}}}  \vdash  \langnt{e_{{\mathrm{2}}}}  \langsym{:}  \tau  \Rightarrow  \Gamma_{{\mathrm{2}}}  \mid   \overline{ C_{{\mathrm{2}}} } 
  }{
    \Theta  \mid  \Gamma  \vdash  \langnt{e_{{\mathrm{1}}}}  \langsym{;}  \langnt{e_{{\mathrm{2}}}}  \langsym{:}  \tau  \Rightarrow  \Gamma_{{\mathrm{2}}}  \mid   \overline{ C_{{\mathrm{1}}} }   \mathbin{+\mkern-10mu+}   \overline{ C_{{\mathrm{2}}} } 
  }

  \vspace{\tyvspace}
  \infrule[T-If]{
    \Theta  \mid  \Gamma  \vdash  \langnt{e_{{\mathrm{1}}}}  \langsym{:}   \texttt{bool}   \Rightarrow  \Gamma'  \mid   \overline{ C_{{\mathrm{1}}} }  \\
    \Theta  \mid  \Gamma'  \vdash  \langnt{e_{{\mathrm{2}}}}  \langsym{:}  \tau  \Rightarrow  \Gamma''  \mid   \overline{ C_{{\mathrm{2}}} } 
    \andalso \Theta  \mid  \Gamma'  \vdash  \langnt{e_{{\mathrm{3}}}}  \langsym{:}  \tau  \Rightarrow  \Gamma''  \mid   \overline{ C_{{\mathrm{3}}} } 
  }{
    \Theta  \mid  \Gamma  \vdash  \langkw{if} \, \langnt{e_{{\mathrm{1}}}} \, \langkw{then} \, \langnt{e_{{\mathrm{2}}}} \, \langkw{else} \, \langnt{e_{{\mathrm{3}}}}  \langsym{:}  \tau  \Rightarrow  \Gamma''  \mid   \overline{ C_{{\mathrm{1}}} }   \mathbin{+\mkern-10mu+}  \langsym{(}     \overline{ C_{{\mathrm{2}}} }    \lor   \overline{ C_{{\mathrm{3}}} }    \langsym{)}
  }

  \vspace{\tyvspace}
  \infrule[T-While]{
    \Theta  \mid  \Gamma  \vdash  \langnt{e_{{\mathrm{1}}}}  \langsym{:}   \texttt{bool}   \Rightarrow  \Gamma  \mid   \overline{ C_{{\mathrm{1}}} } 
    \andalso \Theta  \mid  \Gamma  \vdash  \langnt{e_{{\mathrm{2}}}}  \langsym{:}  \langkw{unit}  \Rightarrow  \Gamma  \mid   \overline{ C_{{\mathrm{2}}} } 
  }{
    \Theta  \mid  \Gamma  \vdash  \langkw{while} \, \langnt{e_{{\mathrm{1}}}} \, \langkw{do} \, \langnt{e_{{\mathrm{2}}}}  \langsym{:}  \langkw{unit}  \Rightarrow  \Gamma  \mid   \langsym{(}   \overline{ C_{{\mathrm{1}}} }   \mathbin{+\mkern-10mu+}   \overline{ C_{{\mathrm{2}}} }   \langsym{)} ^*   \mathbin{+\mkern-10mu+}   \overline{ C_{{\mathrm{1}}} } 
  }

  \vspace{\tyvspace}
  \infrule[T-Call]{
    \Theta  \langsym{(}  \langmv{f}  \langsym{)} =  \Pi   \overline{ \langmv{l'} }   . \braket{ \langmv{x_{{\mathrm{1}}}}  \langsym{:}  \tau_{{\mathrm{1}}}  \langsym{,} \, .. \, \langsym{,}  \langmv{x_{\langmv{n}}}  \langsym{:}  \tau_{\langmv{n}} } \xrightarrow{   \overline{ C }   } \braket{ \Gamma'  |  \tau } 
    \andalso  \sigma_{ \langmv{l} }  = \langsym{[}   \overline{ \langmv{l} }   \slash   \overline{ \langmv{l'} }   \langsym{]}
    \andalso  \overline{ \langmv{l''} }  =  \overline{ \langmv{l} }  \setminus \left(\bigcup_{i=1}^n  \mathrm{flv}(  \sigma_{ \langmv{l} }  \, \tau_{\langmv{i}} ) \right)
  }{
    \Theta  \mid  \Gamma  \langsym{,}   \overline{ \langmv{l''} }   \langsym{,}  \langmv{x_{{\mathrm{1}}}}  \langsym{:}   \sigma_{ \langmv{l} }  \, \tau_{{\mathrm{1}}}  \langsym{,} \, .. \, \langsym{,}  \langmv{x_{\langmv{n}}}  \langsym{:}   \sigma_{ \langmv{l} }  \, \tau_{\langmv{n}}  \vdash  \langmv{f}  \langsym{[}   \overline{ \langmv{l} }   \langsym{]}  \langsym{(}  \langmv{x_{{\mathrm{1}}}}  \langsym{,} \, .. \, \langsym{,}  \langmv{x_{\langmv{n}}}  \langsym{)}  \langsym{:}   \sigma_{ \langmv{l} }  \, \tau  \Rightarrow  \Gamma  \langsym{,}   \sigma_{ \langmv{l} }  \, \Gamma'  \mid   \sigma_{ \langmv{l} }  \, \overline{C}
  }

  \vspace{\tyvspace}
  \infrule[T-FunDecl]{
    \Theta  \vdash  D
    \andalso \Theta  \mid   \overline{ \langmv{l'} }   \langsym{,}  \langmv{x_{{\mathrm{1}}}}  \langsym{:}  \tau_{{\mathrm{1}}}  \langsym{,} \, .. \, \langsym{,}  \langmv{x_{\langmv{n}}}  \langsym{:}  \tau_{\langmv{n}}  \vdash  \langnt{e}  \langsym{:}  \tau  \Rightarrow  \Gamma  \mid  \overline{C}
    \andalso  \overline{ \langmv{l} }  =  \overline{ \langmv{l'} }  \uplus \left(\bigcup_{i=1}^n  \mathrm{flv}( \tau_{\langmv{i}} ) \right)
  }{
    \Theta  \langsym{,}  \langmv{f}  \langsym{:}   \Pi   \overline{ \langmv{l} }   . \braket{ \langmv{x_{{\mathrm{1}}}}  \langsym{:}  \tau_{{\mathrm{1}}}  \langsym{,} \, .. \, \langsym{,}  \langmv{x_{\langmv{n}}}  \langsym{:}  \tau_{\langmv{n}} } \xrightarrow{  \overline{C}  } \braket{ \Gamma  |  \tau }   \vdash  D  \langsym{,}   \langmv{f}  \mapsto [   \overline{ \langmv{l} }   ]( \langmv{x_{{\mathrm{1}}}}  \langsym{,} \, .. \, \langsym{,}  \langmv{x_{\langmv{n}}} ) \langnt{e} 
  }

  \vspace{\tyvspace}
  \infrule[T-Prog]{
    \Theta  \vdash  D
    \andalso \Theta  \mid  \langnt{V}  \langsym{(}  G  \langsym{)}  \vdash  \langnt{e}  \langsym{:}  \tau  \Rightarrow  \Gamma  \mid  \overline{C}
    \andalso G  \mid  \langnt{V}  \langsym{(}  G  \langsym{)}  \vdash  \overline{C}  \Rightarrow  \langmv{L}
  }{
     G   \vdash  \braket{ D ,  \langnt{e} } 
  }

  \caption{Typing rules for classical expressions (excerpt).}
  \label{fig:typing-classical}
\end{figure}

\begin{figure}
  \begin{tabular}{c}
  \begin{minipage}{0.48\columnwidth}
    \infrule[C-Alloc]{
      \langmv{l} \in \langmv{L}
    }{
      G  \mid  \langmv{L}  \vdash  \langkw{alloc} \, \langsym{(}  \langmv{l}  \langsym{)}  \Rightarrow   \langmv{L}  \setminus  \langmv{l} 
    }
  \end{minipage}
  \begin{minipage}{0.48\columnwidth}
  \infax[C-Free]{
    G  \mid  \langmv{L}  \vdash  \langkw{free} \, \langsym{(}  \langmv{l}  \langsym{)}  \Rightarrow  \langmv{L}  \langsym{,}  \langmv{l}
  }
  \end{minipage}
  \end{tabular}

  \vspace{\tyvspace}
  \begin{tabular}{c}
    \begin{minipage}{0.40\columnwidth}
      \infrule[C-Merge]{
        G  \mid  \langmv{L}  \vDash  \langmv{l_{{\mathrm{1}}}}  \sim  \langmv{l_{{\mathrm{2}}}}
      }{
        G  \mid  \langmv{L}  \vdash   \langmv{l_{{\mathrm{1}}}}  \sim  \langmv{l_{{\mathrm{2}}}}   \Rightarrow  \langmv{L}
      }
      
    \end{minipage}
    \begin{minipage}{0.56\columnwidth}
      \infrule[C-If]{
        G  \mid  \langmv{L}  \vdash  \overline{C}_{{\mathrm{1}}}  \Rightarrow  \langmv{L'}
        \andalso G  \mid  \langmv{L}  \vdash  \overline{C}_{{\mathrm{2}}}  \Rightarrow  \langmv{L'}
      }{
        G  \mid  \langmv{L}  \vdash   \overline{C}_{{\mathrm{1}}}  \lor  \overline{C}_{{\mathrm{2}}}   \Rightarrow  \langmv{L'}
      }
    \end{minipage}
  \end{tabular}

  \vspace{\tyvspace}
  \begin{tabular}{c}
    \begin{minipage}{0.32\columnwidth}
      \infrule[C-Loop]{
        G  \mid  \langmv{L}  \vdash  \overline{C}  \Rightarrow  \langmv{L}
      }{
        G  \mid  \langmv{L}  \vdash   \overline{C} ^*   \Rightarrow  \langmv{L}
      }
    \end{minipage}
    \begin{minipage}{0.66\columnwidth}
      \infrule[C-Concat]{
        G  \mid  \langmv{L}  \vdash  \overline{C}_{{\mathrm{1}}}  \Rightarrow  \langmv{L'}
        \andalso G  \mid  \langmv{L'}  \vdash  \overline{C}_{{\mathrm{2}}}  \Rightarrow  \langmv{L''}
      }{
        G  \mid  \langmv{L}  \vdash  \overline{C}_{{\mathrm{1}}}  \mathbin{+\mkern-10mu+}  \overline{C}_{{\mathrm{2}}}  \Rightarrow  \langmv{L''}
      }
    \end{minipage}
  \end{tabular}

  \caption{The rules for connectivity checking.}
  \label{fig:connectivity-rules}
\end{figure}

The typing rules for classical expressions are outlined in \cref{fig:typing-classical}. The rules for creating references, dereferencing, and assigning a value to reference cells closely resemble those of ordinary ML-like languages and are thus omitted here due to space constraints. 

The rule \rn{T-Seq} checks the type of $\langnt{e_{{\mathrm{1}}}}$ and then proceeds to the subsequent expression $\langnt{e_{{\mathrm{2}}}}$ with $\Gamma_{{\mathrm{1}}}$, which is obtained from $\langnt{e_{{\mathrm{1}}}}$, and concatenates two command sequences $ \overline{ C_{{\mathrm{1}}} } $ and $ \overline{ C_{{\mathrm{2}}} } $, which are obtained from $\langnt{e_{{\mathrm{1}}}}$ and $\langnt{e_{{\mathrm{2}}}}$, respectively.

In rule \rn{T-If}, we ensure that each subexpression returns the same type environment $\Gamma'$. This stipulation ensures that regardless of which expression is chosen, the allocation state of the qubits after the if expression is the same. This property facilitates the implementation of an efficient type-checking algorithm, as explained in \cref{sec:type-check}.

In rule \rn{T-While}, we guarantee that the guard expression $\langnt{e_{{\mathrm{1}}}}$ and the body expression $\langnt{e_{{\mathrm{2}}}}$ return the typing environment $\Gamma$ unchanged. This is imperative because the sequential execution of $\langnt{e_{{\mathrm{1}}}}$ and $\langnt{e_{{\mathrm{2}}}}$ must not alter the allocation state of qubits before and after the loop, as the number of loop iterations is unknown.

In rule \rn{T-Call}, we use a substitution map $\sigma_l$ to instantiate occurrences of location variables $\langmv{l}$ in the argument type $\tau_{{\mathrm{1}}}, \dots, \tau_{\langmv{n}}$, return type $\tau$, command sequence $\overline{C}$, and typing environment $\Gamma'$. Prior to calling a function, the caller must ensure that all locations $ \overline{ \langmv{l''} } $ are free, where qubits will be allocated in the function body. After the function call, the portion of the typing environment used within the function transitions to $\sigma_l \, \Gamma'$. Consequently, the rule returns this modified environment along with the unused environment $\Gamma$. Additionally, it returns a command sequence $\sigma_l \, \overline{C}$ generated by the function call.

In rule \rn{T-FunDecl}, akin to let-polymorphism in ML, all free location variables are universally quantified. This enables us to call functions in diverse contexts. Here, location variables $ \overline{ \langmv{l'} } $, which directly appear in the typing environment within the assumption, will be used to allocate new qubits during the execution of the function body. Subsequently, we append the function type, along with the resulting typing environment and command sequence, to a function type environment $\Theta$. We remark that the type of a function is not included in $\Theta$ in the current type checking, as $\langname$ does not support recursive function calls.

The \rn{T-Prog} rule requires that the function definitions in $\Theta$ are well-typed and that the main expression $\langnt{e}$ is also well-typed under $\Theta$ and an initial typing environment $V(G)$, where $G$ represents an architecture graph since all locations are initially unused.

In \rn{T-Prog}, we also check whether $ \overline{ C } $ obtained through the type checking of the main expression is valid under the architecture graph $G$.
The judgment form for the validity of a command sequence is $G  \mid  \langmv{L}  \vdash   \overline{ C }   \Rightarrow  \langmv{L'}$, indicating that the process specified by $ \overline{ C } $ can be safely completed, starting with free locations $\langmv{L}$ under a graph $G$, and that $\langmv{L'}$ are free after executing $\overline{C}$. Formally, we define the rules for the validity of command sequences in \cref{fig:connectivity-rules}.

We prove that any well-typed program under our type system will never encounter a halt in its execution. This assertion is formalized as type soundness, as shown in the following theorem.

\begin{theorem}{(Soundness)}\label{thm:soundness}
  If $ G   \vdash  \braket{ D ,  \langnt{e} } $, then $\rstate{\emptyset}{1}{e}$ does not get stuck.
\end{theorem}

\subsection{Type Checking Algorithm}\label{sec:type-check}

Assuming that $\overline{C}$ is derived from a well-typed expression by type inference, the remaining task involves connectivity checking $G  \mid  \langnt{V}  \langsym{(}  G  \langsym{)}  \vdash  \overline{C}  \Rightarrow  \langmv{L}$ for some $\langmv{L}$.
In this section, we first give a naive algorithm for solving it and then speed it up by reducing the problem to the offline dynamic connectivity problem.

Basically, the problem can be solved by simulating the command sequence from the front to the back. In the simulation, we use the constraints imposed on command sequences in the type system to process the loop and branching commands.
In the case of loops, we can straightforwardly consider $ \overline{C} ^* $ as $\overline{C}$ because the \rn{T-While} rule requires that the locations where qubits are allocated remain unchanged after $\overline{C}$. In other words, the state of locations remains consistent regardless of the number of iterations.

In the case of branches, while we can indeed check both $ \overline{ C_{{\mathrm{1}}} }   \mathbin{+\mkern-10mu+}  \overline{C}$ and $ \overline{ C_{{\mathrm{2}}} }   \mathbin{+\mkern-10mu+}  \overline{C}$ for $\langsym{(}     \overline{ C_{{\mathrm{1}}} }    \lor   \overline{ C_{{\mathrm{2}}} }    \langsym{)}  \mathbin{+\mkern-10mu+}  \overline{C}$ naively, this approach leads to an exponential increase in verification costs for the number of occurrences of non-nested branches. To address this issue, we leverage the fact that the allocation states immediately following $ \overline{ C_{{\mathrm{1}}} } $ and $ \overline{ C_{{\mathrm{2}}} } $ are identical due to the \rn{T-If} rule. Consequently, it suffices to check $ \overline{ C_{{\mathrm{1}}} } $ and $ \overline{ C_{{\mathrm{2}}} }   \mathbin{+\mkern-10mu+}  \overline{C}$ for $\langsym{(}     \overline{ C_{{\mathrm{1}}} }    \lor   \overline{ C_{{\mathrm{2}}} }    \langsym{)}  \mathbin{+\mkern-10mu+}  \overline{C}$, as shown in \cref{lem:trans-branch}.

\begin{lemma}\label{lem:trans-branch}
  Suppose that $\langsym{(}     \overline{ C_{{\mathrm{1}}} }    \lor   \overline{ C_{{\mathrm{2}}} }    \langsym{)}  \mathbin{+\mkern-10mu+}  \overline{C}$ is obtained from a well-typed expression. $G  \mid  \langmv{L}  \vdash  \langsym{(}     \overline{ C_{{\mathrm{1}}} }    \lor   \overline{ C_{{\mathrm{2}}} }    \langsym{)}  \mathbin{+\mkern-10mu+}  \overline{C}  \Rightarrow  \langmv{L'}$ if and only if $G  \mid  \langmv{L}  \vdash   \overline{ C_{{\mathrm{1}}} }   \Rightarrow  \langmv{L''}$ and $G  \mid  \langmv{L}  \vdash   \overline{ C_{{\mathrm{2}}} }   \mathbin{+\mkern-10mu+}  \overline{C}  \Rightarrow  \langmv{L'}$ for some $\langmv{L''}$.
\end{lemma}

The discussion so far gives a type checking algorithm that processes a command sequence in order from the front in a depth-first manner while managing the current allocation state (\cref{alg:naive-type-check} in \cref{app:algorithms}). The time complexity of this approach is $O(E + V)$ for finding a path in a merge command, so the overall operation takes $O(|\overline{C}|(E + V))$ time. 

This algorithm, however, does not scale for large architecture graphs. To address this issue, we reduce the type checking problem to the \emph{offline dynamic connectivity} problem described as follows:
\begin{definition}{(Offline dynamic connectivity)}
  Given a set of vertices $V$ and $Q$ queries $q_1, \dots, q_Q$. Each query is one of the following:
  \begin{itemize}
  \item $\mathtt{add}(u, v)$ : Add an edge $\{u, v\}$ to $G$.
  \item $\mathtt{remove}(u, v)$ : Remove the edge $\{u, v\}$ from $G$.
  \item $\mathtt{connected}(u, v)$ : Answer whether $u$ and $v$ are connected in $G$.
  \end{itemize}
\end{definition}
Dynamic connectivity is known to be efficiently solved~\cite{henzinger1997MaintainingMinimum,holm1998PolylogarithmicDeterministic,thorup2000NearoptimalFullydynamic}, particularly for offline dynamic connectivity, where each query can be processed in $\mathcal{O}(\log V)$ time with a link-cut tree~\cite{sleator1981DataStructureDynamic} maintaining a maximum spanning forest. The rest of this section concentrates on showing how to reduce our connectivity checking problem to the offline dynamic connectivity problem.

The reduction to offline dynamic connectivity is accomplished by converting a command sequence into a single command sequence without branches. Specifically, we transform $   \overline{ C_{{\mathrm{1}}} }    \lor   \overline{ C_{{\mathrm{2}}} }  $ to $  \overline{ C_{{\mathrm{1}}} }    \mathbin{+\mkern-10mu+}  ( \overline{ C_{{\mathrm{1}}} } )^{-1}  \mathbin{+\mkern-10mu+}   \overline{ C_{{\mathrm{2}}} } $, where $  \overline{ C_{{\mathrm{1}}} }  ^{-1} $ is the inverse of $ \overline{ C_{{\mathrm{1}}} } $, an operation that undoes the change in the allocation state made by $ \overline{ C_{{\mathrm{1}}} } $. Roughly speaking, $  \overline{ C }  ^{-1} $ can be obtained by reversing the meaning of $\langkw{alloc} \, \langsym{(}  \langmv{l}  \langsym{)}$ and $\langkw{free} \, \langsym{(}  \langmv{l}  \langsym{)}$ in $ \overline{ C } $ and then reversing the order of $ \overline{ C } $. We can verify $ \overline{ C_{{\mathrm{1}}} } $ and $ \overline{ C_{{\mathrm{2}}} } $ simultaneously by verifying the single command sequence obtained through this conversion. We define this process as the serialization of a command sequence.

\begin{definition}
  For a command sequence $\overline{C}$, the serialized command sequence $ \mathrm{ser}( \overline{C} ) $ is defined by:
  \begin{gather*}
     \mathrm{ser}(  \epsilon  )  =  \epsilon 
    \qquad  \mathrm{ser}(  \langmv{l_{{\mathrm{1}}}}  \sim  \langmv{l_{{\mathrm{2}}}}   \mathbin{+\mkern-10mu+}  \overline{C} )  =  \langmv{l_{{\mathrm{1}}}}  \sim  \langmv{l_{{\mathrm{2}}}}   \mathbin{+\mkern-10mu+}   \mathrm{ser}( \overline{C} )  \\
     \mathrm{ser}( \langkw{alloc} \, \langsym{(}  \langmv{l}  \langsym{)}  \mathbin{+\mkern-10mu+}  \overline{C} )  = \langkw{alloc} \, \langsym{(}  \langmv{l}  \langsym{)}  \mathbin{+\mkern-10mu+}   \mathrm{ser}( \overline{C} )   \mathbin{+\mkern-10mu+}  \langkw{free} \, \langsym{(}  \langmv{l}  \langsym{)} \\
     \mathrm{ser}( \langkw{free} \, \langsym{(}  \langmv{l}  \langsym{)}  \mathbin{+\mkern-10mu+}  \overline{C} )  = \langkw{free} \, \langsym{(}  \langmv{l}  \langsym{)}  \mathbin{+\mkern-10mu+}   \mathrm{ser}( \overline{C} )   \mathbin{+\mkern-10mu+}  \langkw{alloc} \, \langsym{(}  \langmv{l}  \langsym{)} \\
     \mathrm{ser}( \langsym{(}   \overline{C}_{{\mathrm{1}}}  \lor  \overline{C}_{{\mathrm{2}}}   \langsym{)}  \mathbin{+\mkern-10mu+}  \overline{C}_{{\mathrm{3}}} )  =  \mathrm{ser}( \overline{C}_{{\mathrm{1}}} )   \mathbin{+\mkern-10mu+}   \mathrm{ser}( \overline{C}_{{\mathrm{2}}}  \mathbin{+\mkern-10mu+}  \overline{C}_{{\mathrm{3}}} ) 
    \qquad  \mathrm{ser}(  \overline{C} ^*   \mathbin{+\mkern-10mu+}  \overline{C}' )  =  \mathrm{ser}( \overline{C}  \mathbin{+\mkern-10mu+}  \overline{C}' ) 
  \end{gather*}
\end{definition}

\begin{example}
  Consider the following program and $ \overline{ C } $ obtained from the program:

  \begin{tabular}{cc}
    \begin{minipage}[t]{0.45\hsize}
      \begin{lstlisting}[basicstyle=\ttfamily\footnotesize]
let q1 = init(l1) in
let w = mkref
  if true then
    let q2 = init(l2) in
    H(q2);
    let w = meas[Z](q2) in w
  else
    let q2 = init(l2) in
    let w = meas[X, X](q1, q2) in w
in ()
      \end{lstlisting}
    \end{minipage}
    \begin{minipage}[t]{0.49\hsize}
      {\footnotesize
      \begin{align*}
           \overline{ C }  =\ &\langkw{alloc} \, \langsym{(}  \langmv{l_{{\mathrm{1}}}}  \langsym{)}, \\
                  & \langkw{alloc} \, \langsym{(}  \langmv{l_{{\mathrm{2}}}}  \langsym{)}  \lor  \langsym{(}  \langkw{alloc} \, \langsym{(}  \langmv{l_{{\mathrm{2}}}}  \langsym{)}  \langsym{,}   \langmv{l_{{\mathrm{1}}}}  \sim  \langmv{l_{{\mathrm{2}}}}   \langsym{)}  \\
           \mathrm{ser}(  \overline{ C }  )  =\ &\langkw{alloc} \, \langsym{(}  \langmv{l_{{\mathrm{1}}}}  \langsym{)}, \\
                            &\quad \langkw{alloc} \, \langsym{(}  \langmv{l_{{\mathrm{2}}}}  \langsym{)}, \underline{\langkw{free} \, \langsym{(}  \langmv{l_{{\mathrm{2}}}}  \langsym{)}}, \\
                            &\quad \langkw{alloc} \, \langsym{(}  \langmv{l_{{\mathrm{2}}}}  \langsym{)},  \langmv{l_{{\mathrm{1}}}}  \sim  \langmv{l_{{\mathrm{2}}}} , \underline{\langkw{free} \, \langsym{(}  \langmv{l_{{\mathrm{2}}}}  \langsym{)}}, \\
                            &\underline{\langkw{free} \, \langsym{(}  \langmv{l_{{\mathrm{1}}}}  \langsym{)}}
      \end{align*}
      }
    \end{minipage}
  \end{tabular}
  The branch command in $ \overline{ C } $ is serialized as follows:
  \begin{enumerate}
  \item The {\bf then} clause: $\langkw{free} \, \langsym{(}  \langmv{l_{{\mathrm{2}}}}  \langsym{)}$ is inserted after $\langkw{alloc} \, \langsym{(}  \langmv{l_{{\mathrm{2}}}}  \langsym{)}$ to restore the allocation state to one immediately before entering the then clause. As a result, we obtain $\langkw{alloc} \, \langsym{(}  \langmv{l_{{\mathrm{2}}}}  \langsym{)}$, $\underline{\langkw{free} \, \langsym{(}  \langmv{l_{{\mathrm{2}}}}  \langsym{)}}$.
  \item The {\bf else} clause and the subsequent commands: $\langkw{free} \, \langsym{(}  \langmv{l_{{\mathrm{2}}}}  \langsym{)}$ is inserted after these commands because the subsequent expression does not change the allocation state. The merge command $ \langmv{l_{{\mathrm{1}}}}  \sim  \langmv{l_{{\mathrm{2}}}} $ is ignored during inversion. As a result, we obtain $\langkw{alloc} \, \langsym{(}  \langmv{l_{{\mathrm{2}}}}  \langsym{)}  \langsym{,}   \langmv{l_{{\mathrm{1}}}}  \sim  \langmv{l_{{\mathrm{2}}}} , \underline{\langkw{free} \, \langsym{(}  \langmv{l_{{\mathrm{2}}}}  \langsym{)}}$.
  \item Concatenate the results of steps 1 and 2.
  \end{enumerate}
  Finally, the underlined commands in $ \mathrm{ser}(  \overline{ C }  ) $ are inserted as an inverse operation.
\end{example}

Indeed, we can employ $ \mathrm{ser}( \overline{C} ) $ for connectivity checking in place of $\overline{C}$. Moreover, the length of $ \mathrm{ser}( \overline{C} ) $ is less than twice the length of $\overline{C}$, namely $\mathcal{O}(|\overline{C}|)$. We show these properties in the following propositions.

\begin{proposition}\label{prop:ser-sound}
  Suppose that $\overline{C}$ is obtained from a well-typed expression. Then $G  \mid  \langmv{L}  \vdash   \mathrm{ser}( \overline{C} )   \Rightarrow  \langmv{L}$ if and only if $G  \mid  \langmv{L}  \vdash  \overline{C}  \Rightarrow  \langmv{L'}$ for some $\langmv{L'}$.
\end{proposition}

\begin{proposition}\label{prop:ser-length}
  For any $ \overline{ C } $, $| \mathrm{ser}( \overline{C} ) | \leq 2|\overline{C}|$.
\end{proposition}

The final task is to convert $\langkw{alloc} \, \langsym{(}  \langmv{l}  \langsym{)}$ and $\langkw{free} \, \langsym{(}  \langmv{l}  \langsym{)}$, which are operations on vertices, into edge operations in the dynamic connectivity problem. This can be done as follows: when encountering $\langkw{alloc} \, \langsym{(}  \langmv{l}  \langsym{)}$, we remove all edges incident to vertex $\langmv{l}$, and conversely, when encountering $\langkw{free} \, \langsym{(}  \langmv{l}  \langsym{)}$, we restore them. Subsequently, we can determine whether $ \langmv{l_{{\mathrm{1}}}}  \sim  \langmv{l_{{\mathrm{2}}}} $ is valid by checking if they are adjacent to each other or there exist $\langmv{l'_{{\mathrm{1}}}} \in \mathrm{adj}(\langmv{l_{{\mathrm{1}}}})$ and $\langmv{l'_{{\mathrm{2}}}} \in \mathrm{adj}(\langmv{l_{{\mathrm{2}}}})$ such that $\langmv{l'_{{\mathrm{1}}}}$ and $\langmv{l'_{{\mathrm{2}}}}$ are connected. Here, $\mathrm{adj}(l)$ denotes the set of adjacent vertices of $\langmv{l}$ in the initial graph. In fact, if such a path exists, we can use the path $\langmv{l_{{\mathrm{1}}}} \rightarrow \langmv{l'_{{\mathrm{1}}}} \rightarrow \langmv{l'_{{\mathrm{2}}}} \rightarrow \langmv{l_{{\mathrm{2}}}}$ to merge $\langmv{l_{{\mathrm{1}}}}$ and $\langmv{l_{{\mathrm{2}}}}$.

\begin{example}
  Consider a path graph $G$ where $V(G) = \{\langmv{l_{{\mathrm{1}}}}  \langsym{,}  \langmv{l_{{\mathrm{2}}}}  \langsym{,}  \langmv{l_{{\mathrm{3}}}}  \langsym{,}  \langmv{l_{{\mathrm{4}}}}\}$ and $E(G) = \{ \{\langmv{l_{{\mathrm{1}}}}  \langsym{,}  \langmv{l_{{\mathrm{2}}}}\}, \{\langmv{l_{{\mathrm{2}}}}  \langsym{,}  \langmv{l_{{\mathrm{3}}}}\}, \{\langmv{l_{{\mathrm{3}}}}  \langsym{,}  \langmv{l_{{\mathrm{4}}}}\}\}$.
  If $\langkw{alloc} \, \langsym{(}  \langmv{l_{{\mathrm{1}}}}  \langsym{)}$, $\langkw{alloc} \, \langsym{(}  \langmv{l_{{\mathrm{3}}}}  \langsym{)}$, $ \langmv{l_{{\mathrm{1}}}}  \sim  \langmv{l_{{\mathrm{3}}}} $ are obtained, then it is converted into the following queries:
  $\mathtt{remove}(\langmv{l_{{\mathrm{1}}}}  \langsym{,}  \langmv{l_{{\mathrm{2}}}})$, $\mathtt{remove}(\langmv{l_{{\mathrm{2}}}}  \langsym{,}  \langmv{l_{{\mathrm{3}}}})$, $\mathtt{remove}(\langmv{l_{{\mathrm{3}}}}  \langsym{,}  \langmv{l_{{\mathrm{4}}}})$, \underline{$\mathtt{connected}(\langmv{l_{{\mathrm{2}}}}  \langsym{,}  \langmv{l_{{\mathrm{2}}}})$, $\mathtt{connected}(\langmv{l_{{\mathrm{2}}}}  \langsym{,}  \langmv{l_{{\mathrm{4}}}})$}. This sequence is safe if either underlined query is true. In this case, $\mathtt{connected}(\langmv{l_{{\mathrm{2}}}}  \langsym{,}  \langmv{l_{{\mathrm{2}}}})$ returns true and thus it is safe.
\end{example}

The computational complexity of our type checking algorithm for general architecture graphs is $\mathcal{O}(|\overline{C}|V^2\log V)$. This is because the length of $ \mathrm{ser}( \overline{C} ) $ is $\mathcal{O}(|\overline{C}|)$, as indicated by \cref{prop:ser-length}. Additionally, up to $|\mathrm{adj}(l)|^2 = \mathcal{O}(V^2)$ queries are generated for each command\footnote{There is an implementation trick to improve the complexity to $\mathcal{O}(| \overline{ C } |V\log V)$.}. However, in practice, architecture graphs are often sparse, meaning the size of $\mathrm{adj}(l)$ remains at most constant (e.g., 4 in a grid graph). In such cases, the algorithm works in $\mathcal{O}(|\overline{C}|\log V)$.

It is important to note that the size of $ \overline{ C } $ can grow exponentially with respect to the size of a given program, as command sequences are copied in function calls. However, this problem does not occur with loops. Roughly speaking, the size of $ \overline{ C } $ corresponds to the number of surgery operations performed at runtime, assuming that all loops are expanded exactly once.

\section{Extensions}\label{sec:extensions}

\subsection{Recursive Functions}

Currently, $\langname$ lacks support for recursive functions. However, we can introduce support for recursive functions while disregarding the cost of type checking. This section outlines the extension of $\langname$ with recursive functions.

We introduce command variables $\alpha, \beta, \ldots$ into $\overline{C}$ and assign them to functions for type inference. Through type inference, we derive production rules on command sequences denoted as $\alpha \rightarrow \overline{C_\alpha}$, where $\overline{C_\alpha}$ may recursively include $\alpha$ itself. For instance, if a function $f$ has $\langnt{e_{{\mathrm{1}}}}  \langsym{;}  \langmv{f}  \langsym{[}   \overline{ \langmv{l} }   \langsym{]}  \langsym{(}  \langmv{x_{{\mathrm{1}}}}  \langsym{,} \, .. \, \langsym{,}  \langmv{x_{\langmv{n}}}  \langsym{)}  \langsym{;}  \langnt{e_{{\mathrm{2}}}}$ as its body and $\langnt{e_{{\mathrm{1}}}}, \langnt{e_{{\mathrm{2}}}}$ generate command sequences $\overline{C}_{{\mathrm{1}}}, \overline{C}_{{\mathrm{2}}}$ respectively, then we obtain $\alpha \rightarrow \overline{C}_{{\mathrm{1}}} \doubleplus \alpha \doubleplus \overline{C}_{{\mathrm{2}}}$. Consequently, we establish a language $\mathcal{L}$ where each word $\overline{C} \in \mathcal{L}$ corresponds to a potential execution trace of a given program. Our objective is to ascertain whether $\mathcal{L}$ is \emph{safe}, meaning $\forall  \overline{ C }  \in \mathcal{L}.\ \exists \langmv{L}.\ G  \mid  \langnt{V}  \langsym{(}  G  \langsym{)}  \vdash  \overline{C}  \Rightarrow  \langmv{L}$.

We can address this decision problem by constructing the safe language $\mathcal{L}_{\text{safe}} = \mathcal{L}(\mathcal{A})$ from an architecture graph, where $\mathcal{A}$ is a nondeterministic finite automaton created by the following steps:
\begin{enumerate}
  \item Prepare $2^{|V(G)|}$ vertices with empty edges. Each state represents a graph state with cells allocated in different patterns. All states are accepting states.
  \item For each state $u, v \in \mathcal{A}$, add an edge $u \rightarrow v$ if a transition from $u$ to $v$ can occur via $\langkw{alloc} \, \langsym{(}  \langmv{l}  \langsym{)}$ or $\langkw{free} \, \langsym{(}  \langmv{l}  \langsym{)}$ for some $\langmv{l}$.
  \item Add self-loops $u \rightarrow u$ labeled with $ \langmv{l_{{\mathrm{1}}}}  \sim  \langmv{l_{{\mathrm{2}}}} $ if there is a path from $\langmv{l_{{\mathrm{1}}}}$ to $\langmv{l_{{\mathrm{2}}}}$ in the state $u$.
\end{enumerate}

The language $\mathcal{L}_{\text{safe}} = \mathcal{L}(\mathcal{A})$ contains all safe command sequences, and thus we can reduce the decision problem to the model checking problem: $\mathcal{L} \cap \overline{\mathcal{L}_{\text{safe}}} = \emptyset$? However, two primary issues arise. Firstly, efficiently solving this emptiness problem is challenging because $\mathcal{L}$ is generally a context-free language. Secondly, the number of states of $\mathcal{A}$ grows exponentially concerning the number of cells $n$.

One solution to tackle the first issue is constraining recursive function calls. For example, this can be achieved by restricting recursive calls so that the production rules are left-normal or right-normal form and the language $\mathcal{L}$ is normal.

As for the second issue, we currently need a fundamental solution. However, if $|V(G)|$ is not excessively large, this may not present a problem in practice. In fact, it is challenging to increase the number of logical qubits per core significantly.

\subsection{Multi-qubit Quantum Measurements}

The merge operation in lattice surgery can accept three or more qubits, enabling the implementation of multi-body measurements $M_{B_1, \dots, B_n}$ for $n > 2$. Extending $\langname$ with multi-qubit measurements is straightforward by generalizing the binary relation $ \langmv{l_{{\mathrm{1}}}}  \sim  \langmv{l_{{\mathrm{2}}}} $ in our system to an $n$-ary relation $\mathrm{c}(\langmv{l_{{\mathrm{1}}}}  \langsym{,} \, .. \, \langsym{,}  \langmv{l_{\langmv{n}}})$. In this scenario, $\mathrm{c}(\langmv{l_{{\mathrm{1}}}}  \langsym{,} \, .. \, \langsym{,}  \langmv{l_{\langmv{n}}})$ indicates that a program aims to merge $\langmv{l_{{\mathrm{1}}}}  \langsym{,} \, .. \, \langsym{,}  \langmv{l_{\langmv{n}}}$ on a Steiner tree with $\langmv{l_{{\mathrm{1}}}}  \langsym{,} \, .. \, \langsym{,}  \langmv{l_{\langmv{n}}}$ as endpoints, rather than on a path.

Supporting multi-qubit measurements does incur higher costs in type checking due to the increased complexity of the Steiner tree problem compared to pathfinding. For instance, the Dreyfus--Wagner algorithm~\cite{dreyfus1971SteinerProblem}, a widely-known method for finding a minimum Steiner tree, operates in exponential time. However, the Steiner tree for merging qubits does not necessarily need to be minimum, and thus we can employ simpler algorithms to address this task.

\section{Related Work}\label{sec:related-work}

\subsubsection{Verified Quantum Compilation}
One of the most widely studied verifications of quantum compilers is the equivalence checking of quantum circuits, which ensures that the compiler preserves the semantics of the program. Amy developed a framework for reasoning about quantum circuits based on Feynman path integral formalism~\cite{amy2019LargescaleFunctional}. Additionally, several approaches based on decision diagrams have been proposed~\cite{miller2006QMDDDecision,wang2008XQDDBasedVerification,burgholzer2020ImprovedDDbased}.
Another technique is ZX-calculus~\cite{coecke2008InteractingQuantum}, a graphical language with a small set of rewrite rules on string diagrams called ZX diagrams. The completeness of the ZX-calculus~\cite{jeandel2018CompleteAxiomatisation} allows for the transformation of a ZX diagram into any equivalent diagram, making it useful for equivalence checking~\cite{peham2022EquivalenceChecking} and verified optimization~\cite{fagan2019OptimisingClifford,kissinger2020PyZXLarge} of quantum circuits.
While these studies aim to verify the equivalence of quantum programs before and after compilation, our study focuses on verifying whether quantum programs can be executed without runtime errors on quantum devices, considering various hardware constraints. 

Our work is not the only framework aimed at verifying if a quantum program satisfies hardware constraints. Smith and Thornton~\cite{smith2019QuantumComputational} provided a software tool based on Quantum Multiple Decision Diagrams (QMDDs), which synthesizes a logical quantum circuit and maps it to a physical one on a specified architecture. Hietala \etal\cite{hietala2021VerifiedOptimizer} developed {\sc voqc}, a verified optimizer for quantum circuits implemented using the Coq proof assistant, which can verify the correctness of their mapping algorithm that maps qubit variables to physical qubits to handle connectivity constraints.
In contrast to these studies, our work is the first to address the connectivity problem between logical qubits in lattice surgery. 

\subsubsection{Quantum Compilers for Lattice Surgery}
OpenSurgery~\cite{paler2020OpenSurgeryTopological}, presented by Paler and Fowler, compiles Clifford+T circuits naively using the Solovay--Kitaev algorithm. Watkins~\cite{watkins2024HighPerformancea} has recently extended this work, developing a toolchain targeting surface codes and lattice surgery. This extension allows large quantum programs to be compiled efficiently, addressing scalability and performance in fault-tolerant quantum computation. Litinski~\cite{litinski2019GameSurface} presented a Pauli-based approach, which translates a quantum circuit to a sequence of Pauli product measurements of $\pi/8$ and $\pi/4$ rotations and results in removing all Clifford gates from the circuit. Beverland \etal\cite{beverland2022SurfaceCode} developed an efficient algorithm called Edge-Disjoint Paths Compilation (EDPC) to optimize the depth of quantum circuits by parallelizing merge operations. Since our method does not rely on a specific compilation algorithm, we can readily apply our framework to verify the programs produced by these advanced compilers.

\section{Conclusion}\label{sec:conclusion}

We presented a type-based verification method based on $\langname$, which formalizes the execution model of surgery operations. The type system of $\langname$ allows us to extract a command sequence representing the graph operations performed by tracking the positions of individual qubits at the type level. We proved the type soundness by stating that if the resulting command sequence on the architecture graph is safe, the target program will not terminate illegally during execution. Furthermore, we developed an algorithm in this study that efficiently inspects whether the resulting command sequence is safe. Our algorithm achieves this by reducing the type checking problem to the offline dynamic connectivity problem.

Several future challenges remain in this research. Firstly, we intend to extend the methodology of this study to concurrent programs. Optimizing high-level quantum programs and transforming them into $\langname$ programs is also an important task.

\begin{credits}
\subsubsection{\ackname}
This work was supported by JST SPRING, Grant Number JPMJSP2110.
This work is also supported by JST CREST Grant Number JPMJCR23I4, JST Moonshot R\&D Grant Number JPMJMS2061, MEXT Q-LEAP Grant No. JPMXS0120319794, and No. JPMXS0118068682.
\end{credits}

\bibliographystyle{splncs04}
\bibliography{references}

\forarxiv{
\appendix
\section{Appendix}

\subsection{Complete Definition of the Semantics and Type System}

This section provides several definitions, evaluation rules and typing rules omitted in the paper.
Evaluation rules and typing rules are given in \cref{fig:semantics-appendix} and \cref{fig:typing-rules-appendix}, respectively.

\newlength{\evalvspace}
\setlength{\evalvspace}{6pt}

\begin{figure}
  \infrule[]{
    \rstate{ H }{ \rho }{ \langnt{e_{{\mathrm{1}}}} }
    \evalarrow
    \rstate{ H' }{ \rho' }{ \langnt{e'_{{\mathrm{1}}}} }
  }{
    \rstate{ H }{ \rho }{ \langkw{let} \, \langmv{x}  \langsym{=}  \langkw{mkref} \, \langnt{e_{{\mathrm{1}}}} \, \langkw{in} \, \langnt{e_{{\mathrm{2}}}} }
    \evalarrow
    \rstate{ H' }{ \rho' }{ \langkw{let} \, \langmv{x}  \langsym{=}  \langkw{mkref} \, \langnt{e'_{{\mathrm{1}}}} \, \langkw{in} \, \langnt{e_{{\mathrm{2}}}} }
  }

  \vspace{\evalvspace}
  \infrule[]{
    \langmv{x'} \not\in \dom(H)
  }{
    \rstate{ H }{ \rho }{ \langkw{let} \, \langmv{x}  \langsym{=}  \langkw{mkref} \, \langmv{v} \, \langkw{in} \, \langnt{e} }
    \evalarrow
    \rstate{ H\{ x' \mapsto v \} }{ \rho }{ \langsym{[}  \langmv{x'}  \slash  \langmv{x}  \langsym{]} \, \langnt{e} }
  }

  \vspace{\evalvspace}
  \infrule[]{
    \langmv{x} \in \dom(H)
  }{
    \rstate{ H }{ \rho }{ \langsym{*}  \langmv{x} }
    \evalarrow
    \rstate{ H }{ \rho }{ H(x) }
  }

  \vspace{\evalvspace}
  \infrule[]{
    \rstate{ H }{ \rho }{ \langnt{e_{{\mathrm{1}}}} }
    \evalarrow
    \rstate{ H' }{ \rho' }{ \langnt{e'_{{\mathrm{1}}}} }
  }{
    \rstate{ H }{ \rho }{ \langnt{e_{{\mathrm{1}}}}  \langsym{;}  \langnt{e_{{\mathrm{2}}}} }
    \evalarrow
    \rstate{ H' }{ \rho' }{ \langnt{e'_{{\mathrm{1}}}}  \langsym{;}  \langnt{e_{{\mathrm{2}}}} }
  }

  \vspace{\evalvspace}
  \infax[]{
    \rstate{ H }{ \rho }{ \langsym{(}  \langsym{)}  \langsym{;}  \langnt{e} }
    \evalarrow
    \rstate{ H }{ \rho }{ \langnt{e} }
  }

  \vspace{\evalvspace}
  \infrule[]{
    \rstate{ H }{ \rho }{ \langnt{e_{{\mathrm{1}}}} }
    \evalarrow
    \rstate{ H }{ \rho }{ \langnt{e'_{{\mathrm{1}}}} }
  }{
    \rstate{ H }{ \rho }{ \langkw{if} \, \langnt{e_{{\mathrm{1}}}} \, \langkw{then} \, \langnt{e_{{\mathrm{2}}}} \, \langkw{else} \, \langnt{e_{{\mathrm{3}}}} }
    \evalarrow
    \rstate{ H }{ \rho }{ \langkw{if} \, \langnt{e'_{{\mathrm{1}}}} \, \langkw{then} \, \langnt{e_{{\mathrm{2}}}} \, \langkw{else} \, \langnt{e_{{\mathrm{3}}}} }
  }

  \vspace{\evalvspace}
  \infax[]{
    \rstate{ H }{ \rho }{ \langkw{if} \,  \texttt{true}  \, \langkw{then} \, \langnt{e_{{\mathrm{1}}}} \, \langkw{else} \, \langnt{e_{{\mathrm{2}}}} }
    \evalarrow 
    \rstate{ H }{ \rho }{ \langnt{e_{{\mathrm{1}}}} }
  }

  \vspace{\evalvspace}
  \infax[]{
    \rstate{ H }{ \rho }{ \langkw{if} \,  \texttt{false}  \, \langkw{then} \, \langnt{e_{{\mathrm{1}}}} \, \langkw{else} \, \langnt{e_{{\mathrm{2}}}} }
    \evalarrow 
    \rstate{ H }{ \rho }{ \langnt{e_{{\mathrm{2}}}} }
  }

  \vspace{\evalvspace}
  \infax[]{
      \left[  \langmv{H} ,  \rho ,  \langkw{while} \, \langnt{e_{{\mathrm{1}}}} \, \langkw{do} \, \langnt{e_{{\mathrm{2}}}}  \right]   \rightarrow_{ \langmv{D} ,  G }   \left[  \langmv{H} ,  \rho ,  \langkw{if} \, \langnt{e_{{\mathrm{1}}}} \, \langkw{then} \, \langnt{e_{{\mathrm{2}}}}  \langsym{;}  \langsym{(}  \langkw{while} \, \langnt{e_{{\mathrm{1}}}} \, \langkw{do} \, \langnt{e_{{\mathrm{2}}}}  \langsym{)} \, \langkw{else} \, \langsym{(}  \langsym{)}  \right]  
  }

  \caption{Operational semantics.}
  \label{fig:semantics-appendix}
\end{figure}

\begin{figure}[tb]
  \infrule[T-MkRef]{
    \Theta  \mid  \Gamma  \vdash  \langnt{e_{{\mathrm{1}}}}  \langsym{:}  \tau_{{\mathrm{1}}}  \Rightarrow  \Gamma'  \mid   \overline{ C_{{\mathrm{1}}} } 
    \andalso  \mathrm{flv}( \tau_{{\mathrm{1}}} )  = \emptyset \\
    \Theta  \mid  \Gamma'  \langsym{,}  \langmv{x}  \langsym{:}  \langkw{ref} \, \tau_{{\mathrm{1}}}  \vdash  \langnt{e_{{\mathrm{2}}}}  \langsym{:}  \tau_{{\mathrm{2}}}  \Rightarrow  \Gamma''  \mid   \overline{ C_{{\mathrm{2}}} } 
    \andalso \langmv{x} \not\in \dom(\Gamma'')
  }{
    \Theta  \mid  \Gamma  \vdash  \langkw{let} \, \langmv{x}  \langsym{=}  \langkw{mkref} \, \langnt{e_{{\mathrm{1}}}} \, \langkw{in} \, \langnt{e_{{\mathrm{2}}}}  \langsym{:}  \tau_{{\mathrm{2}}}  \Rightarrow  \Gamma''  \mid   \overline{ C_{{\mathrm{1}}} }   \mathbin{+\mkern-10mu+}   \overline{ C_{{\mathrm{2}}} } 
  }

  \vspace{\tyvspace}
  \infrule[T-Deref]{
    \Theta  \mid  \Gamma  \vdash  \langnt{e}  \langsym{:}  \langkw{ref} \, \tau  \Rightarrow  \Gamma'  \mid  \overline{C}
  }{
    \Theta  \mid  \Gamma  \vdash  \langsym{*}  \langnt{e}  \langsym{:}  \tau  \Rightarrow  \Gamma'  \mid  \overline{C}
  }
  
  \vspace{\tyvspace}
  \infrule[T-Assign]{
    \langmv{x}  \langsym{:}  \langkw{ref} \, \tau \in \Gamma
    \andalso \Theta  \mid  \Gamma  \vdash  \langnt{e}  \langsym{:}  \tau  \Rightarrow  \Gamma'  \mid  \overline{C}
  }{
    \Theta  \mid  \Gamma  \vdash  \langmv{x}  \langsym{:=}  \langnt{e}  \langsym{:}  \langkw{unit}  \Rightarrow  \Gamma'  \mid  \overline{C}
  }

  \vspace{\tyvspace}
  \begin{tabular}{c}
    \begin{minipage}{0.48\columnwidth}
      \infrule[S-TyEnv]{
        \{ l \in \Gamma \} = \{ l \in \Gamma' \}
        \andalso \Gamma' \subseteq \Gamma
      }{
         \Gamma  \leq  \Gamma' 
      }
    \end{minipage}
    \begin{minipage}{0.48\columnwidth}
      \infrule[T-Sub]{
         \Gamma'  \leq  \Gamma'' 
        \andalso \Theta  \mid  \Gamma  \vdash  \langnt{e}  \langsym{:}  \tau  \Rightarrow  \Gamma'  \mid  \overline{C}
      }{
        \Theta  \mid  \Gamma  \vdash  \langnt{e}  \langsym{:}  \tau  \Rightarrow  \Gamma''  \mid  \overline{C}
      }
    \end{minipage}
  \end{tabular}
  \caption{The typing rules omitted in the main text.}
  \label{fig:typing-rules-appendix}
\end{figure}

We also define a typing rule for runtime states in \cref{fig:typing-rules-rs} to prove our typing preservation lemma.

\begin{figure}[tb]
  \infrule[T-RState]{
    \dom(\Gamma) \subseteq \dom(H)
    \andalso  \mathrm{Used}( \Gamma )  \subseteq \Used(H)
    \andalso  \mathrm{flv}( \Gamma )  \subseteq V(G) \\
    \Theta  \vdash  \langmv{D}
    \andalso \Theta  \mid  \Gamma  \vdash  \langnt{e}  \langsym{:}  \tau  \Rightarrow  \Gamma'  \mid  \overline{C}
    \andalso G  \mid   \langnt{V}  \langsym{(}  G  \langsym{)}  \setminus   \mathrm{Used}( \langmv{H} )    \vdash  \overline{C}  \Rightarrow  \langmv{L}
  }{
    \Theta  \mid  \Gamma  \mid  \langmv{D}  \mid  G  \vdash   \left[  \langmv{H} ,  \rho ,  \langnt{e}  \right]   \Rightarrow  \Gamma'  \mid  \langmv{L}
  }

  \begin{align*}
     \mathrm{Used}( \Gamma )  &\coloneqq \{\langmv{l} \mid \exists \langmv{x}  \langsym{:}  \tau \in \Gamma.\ \langmv{l} \in  \mathrm{flv}( \tau )  \} \\
    \Used(H) &\coloneqq \{l \mid \exists x \in \dom(H).\ H(x) = l\}
  \end{align*}
  \caption{The typing rule for runtime states.}
  \label{fig:typing-rules-rs}
\end{figure}

\begin{definition}\label{def:length-cs}
  For a command sequence $\overline{C}$, $|\overline{C}|$ denotes the size of $\overline{C}$ which are defined by:
  \begin{gather*}
    | \epsilon | = 0 \\
    |\langkw{alloc} \, \langsym{(}  \langmv{l}  \langsym{)}  \mathbin{+\mkern-10mu+}  \overline{C}| = |\langkw{free} \, \langsym{(}  \langmv{l}  \langsym{)}  \mathbin{+\mkern-10mu+}  \overline{C}| = | \langmv{l_{{\mathrm{1}}}}  \sim  \langmv{l_{{\mathrm{2}}}}   \mathbin{+\mkern-10mu+}  \overline{C}| = 1 + |\overline{C}| \\
    |   \overline{ C_{{\mathrm{1}}} }    \lor   \overline{ C_{{\mathrm{2}}} }   \mathbin{+\mkern-10mu+}  \overline{C} | = 1 + | \overline{ C_{{\mathrm{1}}} } | + | \overline{ C_{{\mathrm{2}}} } | + |\overline{C}| \\
    |  \overline{ C }  ^*   \mathbin{+\mkern-10mu+}   \overline{ C' } | = 1 + | \overline{ C } | + | \overline{ C' } |
  \end{gather*}
\end{definition}

\subsection{Proofs}

This sections provides the proofs of lemmas for type soundness. The main lemma is the progress lemma (\cref{lem:progress}) and the type preservation lemma (\cref{lem:type-preserve}).

\begin{lemma}\label{lem:ty-initial-rstate}
  If $ G   \vdash  \braket{ D ,  \langnt{e} } $, then $\Theta \mid  V( G )  \mid \langmv{D} \mid G \vdash \left[\emptyset, 1, \langnt{e}\right]$.
\end{lemma}

\begin{proof}
  By inversion on \rn{T-Prog}, we have $\Theta  \vdash  \langmv{D}$, $\Theta  \mid  \langnt{V}  \langsym{(}  G  \langsym{)}  \vdash  \langnt{e}  \langsym{:}  \tau  \Rightarrow  \Gamma  \mid  \overline{C}$ and $G  \mid  \langnt{V}  \langsym{(}  G  \langsym{)}  \vdash  \overline{C}  \Rightarrow  \langmv{L}$.
  Now $ \mathrm{Used}( \Gamma )  =  \mathrm{Used}( \langmv{H} )  = \emptyset$ so $ \mathrm{Used}( \Gamma )  \subseteq  \mathrm{Used}( \langmv{H} ) $.
  Moreover, $ \mathrm{flv}( \Gamma )  =  V( G ) $ and thus $ \mathrm{flv}( \Gamma )  \subseteq  V( G ) $.
  Consequently we have $\Theta \mid  V( G )  \mid \langmv{D} \mid G \vdash \left[\emptyset, 1, \langnt{e}\right]$.
\end{proof}

\begin{lemma}{(Progress)}\label{lem:progress}
  If $\Theta  \mid  \Gamma  \mid  \langmv{D}  \mid  G  \vdash   \left[  \langmv{H} ,  \rho ,  \langnt{e}  \right]   \Rightarrow  \Gamma'  \mid  \langmv{L}$, then $\langnt{e}$ is a value
  or there exists $ \left[  \langmv{H'} ,  \rho' ,  \langnt{e'}  \right] $ such that $  \left[  \langmv{H} ,  \rho ,  \langnt{e}  \right]   \rightarrow_{ \langmv{D} ,  G }   \left[  \langmv{H'} ,  \rho' ,  \langnt{e'}  \right]  $.
\end{lemma}

\begin{proof}
  Prove by case analysis of $\langnt{e}$.
  \pfcase{$\langnt{e} = x$}
    $\langmv{x}$ is a value an thus the statement follows immediately.

  \pfcase{$\langnt{e} = \langkw{let} \, \langmv{x}  \langsym{=}  \langkw{init} \, \langsym{(}  \langmv{l}  \langsym{)} \, \langkw{in} \, \langnt{e'}$}
    By \rn{T-RState}, we have $\langmv{l} \in  \mathrm{flv}( \Gamma )  \subseteq G$ and $\langmv{l} \not\in \cod(H)$. Then the state can step by \rn{E-Init}.

  \pfcase{$\langnt{e} = \langkw{let} \, \langmv{x}  \langsym{=}  \langkw{minit} \, \langsym{(}  \langmv{l}  \langsym{)} \, \langkw{in} \, \langnt{e'}$}
    Similar to the $\langkw{init} \, \langsym{(}  \langmv{l}  \langsym{)}$ case above.

  \pfcase{$\langnt{e} = \langkw{free} \, \langmv{x}  \langsym{;}  \langnt{e'}$}
    By inversion on $\Theta  \mid  \Gamma  \vdash  \langkw{free} \, \langmv{x}  \langsym{;}  \langnt{e'}  \langsym{:}  \tau  \Rightarrow  \Gamma''  \mid  \overline{C}$, we have $\Gamma = \Gamma'  \langsym{,}  \langmv{l}$ and $\tau =  \texttt{qbit}( \langmv{l} ) $ for some $\Gamma'  \langsym{,}  \Gamma''$.
    Then $\langmv{l} \in \Gamma \subseteq \cod(H)$ by \rn{T-RState} and thus the state can step by \rn{E-Free}.

  \pfcase{$\langnt{e} = \langkw{let} \, \langmv{x}  \langsym{=}   M_{ \langmv{B_{{\mathrm{1}}}}  \langsym{,}  \langmv{B_{{\mathrm{2}}}} }   \langsym{(}  \langmv{x_{{\mathrm{1}}}}  \langsym{,}  \langmv{x_{{\mathrm{2}}}}  \langsym{)} \, \langkw{in} \, \langnt{e'}$}
    By inversion on $\Theta  \mid  \Gamma  \vdash  \langkw{let} \, \langmv{x}  \langsym{=}   M_{ \langmv{B_{{\mathrm{1}}}}  \langsym{,}  \langmv{B_{{\mathrm{2}}}} }   \langsym{(}  \langmv{x_{{\mathrm{1}}}}  \langsym{,}  \langmv{x_{{\mathrm{2}}}}  \langsym{)} \, \langkw{in} \, \langnt{e'}  \langsym{:}  \tau  \Rightarrow  \Gamma'  \mid  \overline{C}$, we have $\langmv{x_{\langmv{i}}}  \langsym{:}   \texttt{qbit}( \langmv{l_{\langmv{i}}} ) $ for $1 \leq i \leq 2$ and $\overline{C} =  \langmv{l_{{\mathrm{1}}}}  \sim  \langmv{l_{{\mathrm{2}}}}   \mathbin{+\mkern-10mu+}   \overline{ C' } $.
    Then $\langmv{l_{{\mathrm{1}}}}  \langsym{,}  \langmv{l_{{\mathrm{2}}}} \in  \mathrm{Used}( \Gamma )  \subseteq \cod(H)$ and thus the state can step by \rn{E-Meas2}.

  The other cases are trivial.
\end{proof}

\begin{lemma}\label{lem:substitution-var}
  If $\Theta  \mid  \Gamma  \langsym{,}  \langmv{x}  \langsym{:}  \tau'  \vdash  \langnt{e}  \langsym{:}  \tau  \Rightarrow  \Gamma'  \mid  \overline{C}$ and $\langmv{x'} \not\in \dom(\Gamma)$,
  then $\Theta  \mid  \Gamma  \langsym{,}  \langmv{x'}  \langsym{:}  \tau'  \vdash  \langsym{[}  \langmv{x'}  \slash  \langmv{x}  \langsym{]} \, \langnt{e}  \langsym{:}  \tau  \Rightarrow  \langsym{[}  \langmv{x'}  \slash  \langmv{x}  \langsym{]} \, \Gamma'  \mid  \overline{C}$.
\end{lemma}

\begin{proof}
  Prove by straightforward induction on the typing derivation.
\end{proof}

\begin{lemma}\label{lem:substitution-locs}
  Suppose that $\Theta  \mid  \Gamma  \vdash  \langnt{e}  \langsym{:}  \tau  \Rightarrow  \Gamma'  \mid   \overline{ C' } $ and $ \overline{ \langmv{l} } $ contains distinct location variables, and $ \mathrm{flv}( \Gamma )  \cap  \overline{ \langmv{l} }  = \emptyset$.
  Then $\Theta  \mid   \sigma_{ \langmv{l} }  \, \Gamma  \vdash   \sigma_{ \langmv{l} }  \, \langnt{e}  \langsym{:}   \sigma_{ \langmv{l} }  \, \tau  \Rightarrow   \sigma_{ \langmv{l} }  \, \Gamma'  \mid   \sigma_{ \langmv{l} }  \,  \overline{ C' } $, where $ \sigma_{ \langmv{l} }  = \langsym{[}   \overline{ \langmv{l} }   \slash   \overline{ \langmv{l'} }   \langsym{]}$.
\end{lemma}

\begin{proof}
  Straightforward induction on the typing derivation.
\end{proof}

\begin{lemma}\label{lem:C-remove-loop}
  If $G  \mid  \langmv{L}  \vdash    \overline{ C_{{\mathrm{1}}} }  ^*   \mathbin{+\mkern-10mu+}   \overline{ C_{{\mathrm{2}}} }   \Rightarrow  \langmv{L'}$, then $G  \mid  \langmv{L}  \vdash   \overline{ C_{{\mathrm{2}}} }   \Rightarrow  \langmv{L'}$.
\end{lemma}

\begin{proof}
  By inversion on \rn{C-Concat}, we have
  \begin{gather*}
    G  \mid  \langmv{L}  \vdash    \overline{ C_{{\mathrm{1}}} }  ^*   \Rightarrow  \langmv{L''}
    \andalso G  \mid  \langmv{L''}  \vdash   \overline{ C_{{\mathrm{2}}} }   \Rightarrow  \langmv{L'}
  \end{gather*}
  for some $\langmv{L''}$.
  Here, by \rn{C-Loop}, $\langmv{L''} = \langmv{L}$.
  Therefore $G  \mid  \langmv{L}  \vdash   \overline{ C_{{\mathrm{2}}} }   \Rightarrow  \langmv{L'}$.
\end{proof}

\begin{lemma}\label{lem:C-loop-fold}
  If $G  \mid  \langmv{L}  \vdash    \overline{ C_{{\mathrm{1}}} }  ^*   \Rightarrow  \langmv{L}$, then $G  \mid  \langmv{L}  \vdash    \overline{ C_{{\mathrm{1}}} }   \mathbin{+\mkern-10mu+}   \overline{ C_{{\mathrm{1}}} }  ^*   \Rightarrow  \langmv{L}$.
\end{lemma}

\begin{proof}
  Straightforward.
\end{proof}

\begin{lemma}\label{lem:C-inversion-concat-arbitrary}
  $G  \mid  \langmv{L}  \vdash   \overline{ C }   \Rightarrow  \langmv{L'}$ if and only if $G  \mid  \langmv{L}  \vdash   \overline{ C_{{\mathrm{1}}} }   \Rightarrow  \langmv{L''}$ and $G  \mid  \langmv{L''}  \vdash   \overline{ C_{{\mathrm{2}}} }   \Rightarrow  \langmv{L}$ for all $ \overline{ C_{{\mathrm{1}}} } ,  \overline{ C_{{\mathrm{2}}} } $ such that $ \overline{ C }  =  \overline{ C_{{\mathrm{1}}} }   \mathbin{+\mkern-10mu+}   \overline{ C_{{\mathrm{2}}} } $.
\end{lemma}

\begin{proof}
  Straightforward induction on the derivation of $G  \mid  \langmv{L}  \vdash   \overline{ C }   \Rightarrow  \langmv{L'}$.
\end{proof}

\begin{lemma}\label{lem:C-branch-distribute}
  $G  \mid  \langmv{L}  \vdash   \overline{ C_{{\mathrm{1}}} }   \mathbin{+\mkern-10mu+}  \langsym{(}     \overline{ C_{{\mathrm{2}}} }    \lor   \overline{ C_{{\mathrm{3}}} }    \langsym{)}  \Rightarrow  \langmv{L'}$
  if and only if $G  \mid  \langmv{L}  \vdash   \overline{ C_{{\mathrm{1}}} }   \mathbin{+\mkern-10mu+}   \overline{ C_{{\mathrm{2}}} }   \Rightarrow  \langmv{L'}$ and $G  \mid  \langmv{L}  \vdash   \overline{ C_{{\mathrm{1}}} }   \mathbin{+\mkern-10mu+}   \overline{ C_{{\mathrm{3}}} }   \Rightarrow  \langmv{L'}$.
\end{lemma}

\begin{proof}
  ($\Rightarrow$)
  By \cref{lem:C-inversion-concat-arbitrary}, we have $G  \mid  \langmv{L}  \vdash   \overline{ C_{{\mathrm{1}}} }   \Rightarrow  \langmv{L''}$ and $G  \mid  \langmv{L''}  \vdash     \overline{ C_{{\mathrm{2}}} }    \lor   \overline{ C_{{\mathrm{3}}} }    \Rightarrow  \langmv{L'}$ for some $\langmv{L''}$.
  By inversion on \rn{C-If}, we have $G  \mid  \langmv{L''}  \vdash   \overline{ C_{{\mathrm{2}}} }   \Rightarrow  \langmv{L'}$ and $G  \mid  \langmv{L''}  \vdash   \overline{ C_{{\mathrm{3}}} }   \Rightarrow  \langmv{L'}$.
  Therefore, we have $G  \mid  \langmv{L}  \vdash   \overline{ C_{{\mathrm{1}}} }   \mathbin{+\mkern-10mu+}   \overline{ C_{{\mathrm{2}}} }   \Rightarrow  \langmv{L'}$ and $G  \mid  \langmv{L}  \vdash   \overline{ C_{{\mathrm{1}}} }   \mathbin{+\mkern-10mu+}   \overline{ C_{{\mathrm{3}}} }   \Rightarrow  \langmv{L'}$ by \rn{C-Concat}.

  ($\Leftarrow$)
  By \cref{lem:C-inversion-concat-arbitrary}, we have
  \begin{gather*}
    G  \mid  \langmv{L}  \vdash   \overline{ C_{{\mathrm{1}}} }   \Rightarrow  \langmv{L''}
    \andalso G  \mid  \langmv{L''}  \vdash   \overline{ C_{{\mathrm{2}}} }   \Rightarrow  \langmv{L'}
    \andalso G  \mid  \langmv{L''}  \vdash   \overline{ C_{{\mathrm{3}}} }   \Rightarrow  \langmv{L'}.
  \end{gather*}
  By \rn{C-If}, we have $G  \mid  \langmv{L''}  \vdash     \overline{ C_{{\mathrm{2}}} }    \lor   \overline{ C_{{\mathrm{3}}} }    \Rightarrow  \langmv{L'}$.
  Finally we have $G  \mid  \langmv{L}  \vdash   \overline{ C_{{\mathrm{1}}} }   \mathbin{+\mkern-10mu+}  \langsym{(}     \overline{ C_{{\mathrm{2}}} }    \lor   \overline{ C_{{\mathrm{3}}} }    \langsym{)}  \Rightarrow  \langmv{L'}$ by \rn{C-Concat}.
\end{proof}

\begin{lemma}{(Weakening)}\label{lem:weakening-tyenv}
  If $\Theta  \mid  \Gamma  \vdash  \langnt{e}  \langsym{:}  \tau  \Rightarrow  \Gamma'  \mid   \overline{ C } $ and $\Gamma'' \cap \Gamma = \emptyset$,
  then $\Theta  \mid  \Gamma  \langsym{,}  \Gamma''  \vdash  \langnt{e}  \langsym{:}  \tau  \Rightarrow  \Gamma'  \langsym{,}  \Gamma''  \mid   \overline{ C } $.
\end{lemma}

\begin{proof}
  Straightforward induction on the typing derivation.
\end{proof}

\begin{lemma}{(Weakening2)}\label{lem:weakening-ftyenv}
  If $\Theta  \mid  \Gamma  \vdash  \langnt{e}  \langsym{:}  \tau  \Rightarrow  \Gamma'  \mid   \overline{ C } $ and $\Theta \subseteq \Theta'$,
  then $\Theta'  \mid  \Gamma  \vdash  \langnt{e}  \langsym{:}  \tau  \Rightarrow  \Gamma'  \mid   \overline{ C } $.
\end{lemma}

\begin{proof}
  Straightforward induction on the typing derivation.
\end{proof}

\begin{lemma}\label{lem:inv-ty-fundecl}
  If $\Theta  \vdash  \langmv{D}$ and $\Theta  \langsym{(}  \langmv{f}  \langsym{)} =  \Pi   \overline{ \langmv{l} }   . \braket{ \langmv{x_{{\mathrm{1}}}}  \langsym{:}  \tau_{\langmv{n}}  \langsym{,} \, .. \, \langsym{,}  \langmv{x_{\langmv{n}}}  \langsym{:}  \tau_{\langmv{n}} } \xrightarrow{   \overline{ C }   } \braket{ \Gamma  |  \tau }  \in \Theta$,
  then there exists $\Theta' \subset \Theta$ and $\langmv{D'} \subset \langmv{D}$ and $\langmv{L}$ such that
  $\Theta'  \vdash  \langmv{D'}$ and $\Theta'  \mid  \langmv{L}  \langsym{,}  \langmv{x_{{\mathrm{1}}}}  \langsym{:}  \tau_{{\mathrm{1}}}  \langsym{,} \, .. \, \langsym{,}  \langmv{x_{\langmv{n}}}  \langsym{:}  \tau_{\langmv{n}}  \vdash  \langnt{e}  \langsym{:}  \tau  \Rightarrow  \Gamma  \mid   \overline{ C } $ and $ \overline{ \langmv{l} }  = \langmv{L} \uplus (\bigcup_{i=1}^n  \mathrm{flv}( \tau_{\langmv{i}} ) )$.
\end{lemma}

\begin{proof}
  Straightforward induction on $\Theta  \vdash  \langmv{D}$.
\end{proof}

\begin{lemma}{(Type preservation)}\label{lem:type-preserve}
  If $\Theta  \mid  \Gamma  \mid  \langmv{D}  \mid  G  \vdash   \left[  \langmv{H} ,  \rho ,  \langnt{e}  \right]   \Rightarrow  \Gamma'  \mid  \langmv{L}$ and $  \left[  \langmv{H} ,  \rho ,  \langnt{e}  \right]   \rightarrow_{ \langmv{D} ,  G }   \left[  \langmv{H'} ,  \rho' ,  \langnt{e'}  \right]  $,
  then there exist $\Gamma'$ such that $\Theta  \mid  \Gamma''  \mid  \langmv{D}  \mid  G  \vdash   \left[  \langmv{H'} ,  \rho' ,  \langnt{e'}  \right]   \Rightarrow  \Gamma'  \mid  \langmv{L}$.
\end{lemma}

\begin{proof}
  We prove by the induction on $\rstate{H}{\rho}{e} \evalarrow \rstate{H'}{\rho'}{e'}$. 
  \pfcaselabel{\rn{E-Init}}
    \begin{gather*}
      \Theta  \mid  \Gamma  \mid  \langmv{D}  \mid  G  \vdash   \left[  \langmv{H} ,  \rho ,  \langkw{let} \, \langmv{x}  \langsym{=}  \langkw{init} \, \langsym{(}  \langmv{l}  \langsym{)} \, \langkw{in} \, \langnt{e}  \right]   \Rightarrow  \Gamma'  \mid  \langmv{L} \\
       \left[  \langmv{H} ,  \rho ,  \langkw{let} \, \langmv{x}  \langsym{=}  \langkw{init} \, \langsym{(}  \langmv{l}  \langsym{)} \, \langkw{in} \, \langnt{e}  \right] 
      \evalarrow
      \rstate{ H\{ x' \mapsto l \} }{ \rho \otimes \lblketbra{0}{0}{l} }{ \langsym{[}  \langmv{x'}  \slash  \langmv{x}  \langsym{]} \, \langnt{e} } \\
      \langmv{l} \not\in \cod(H)
      \andalso \langmv{x'} \not\in \dom(H)
      \andalso \langmv{l} \in G
    \end{gather*}
    By inversion on \rn{T-RState}, we have
    \begin{gather*}
      \dom(\Gamma  \langsym{,}  \langmv{l}) \subseteq \dom(H)
      \andalso  \mathrm{Used}( \Gamma  \langsym{,}  \langmv{l} )  \subseteq  \mathrm{Used}( \langmv{H} ) 
      \andalso  \mathrm{flv}( \Gamma )  \subseteq V(G) \\
      \Theta  \vdash  \langmv{D}
      \andalso \Theta  \mid  \Gamma  \langsym{,}  \langmv{l}  \vdash  \langkw{let} \, \langmv{x}  \langsym{=}  \langkw{init} \, \langsym{(}  \langmv{l}  \langsym{)} \, \langkw{in} \, \langnt{e}  \langsym{:}  \tau  \Rightarrow  \Gamma'  \mid  \langkw{alloc} \, \langsym{(}  \langmv{l}  \langsym{)}  \mathbin{+\mkern-10mu+}   \overline{ C }  \\
      G  \mid   \langnt{V}  \langsym{(}  G  \langsym{)}  \setminus   \mathrm{Used}( \langmv{H} )    \vdash  \langkw{alloc} \, \langsym{(}  \langmv{l}  \langsym{)}  \mathbin{+\mkern-10mu+}   \overline{ C }   \Rightarrow  \langmv{L}
    \end{gather*}
    By inversion on \rn{T-Init}, we have
    \begin{gather*}
      \Theta  \mid  \Gamma  \langsym{,}  \langmv{x}  \langsym{:}   \texttt{qbit}( \langmv{l} )   \vdash  \langnt{e}  \langsym{:}  \tau  \Rightarrow  \Gamma'  \mid   \overline{ C } 
      \andalso x \not\in \dom(\Gamma')
    \end{gather*}
    for some $\Gamma'$.
    Now $\langmv{x'} \not\in \dom(\Gamma)$ because $\langmv{x'} \not\in \dom(H)$ and $\dom(\Gamma  \langsym{,}  \langmv{l}) \subseteq \dom(H)$.
    Therefore we have $\Theta  \mid  \Gamma  \langsym{,}  \langmv{x'}  \langsym{:}   \texttt{qbit}( \langmv{l} )   \vdash  \langsym{[}  \langmv{x'}  \slash  \langmv{x}  \langsym{]} \, \langnt{e}  \langsym{:}  \tau  \Rightarrow  \langsym{[}  \langmv{x'}  \slash  \langmv{x}  \langsym{]} \, \Gamma'  \mid   \overline{ C } $ by \cref{lem:substitution-var}.
    Also $\dom(\Gamma  \langsym{,}  \langmv{x'}  \langsym{:}   \texttt{qbit}( \langmv{l} ) ) \subseteq \dom(H\{x' \mapsto l\})$ because $\dom(\Gamma) \subseteq \dom(H)$.
    Similarly we can show $ \mathrm{Used}( \Gamma  \langsym{,}  \langmv{x'}  \langsym{:}   \texttt{qbit}( \langmv{l} )  )  \subseteq  \mathrm{Used}(  \langmv{H} \{ \langmv{x'}  \mapsto  \langmv{l} \}  ) $.

    Next we will show $G  \mid   \langnt{V}  \langsym{(}  G  \langsym{)}  \setminus   \mathrm{Used}(  \langmv{H} \{ \langmv{x'}  \mapsto  \langmv{l} \}  )    \vdash   \overline{ C }   \Rightarrow  \langmv{L}$.
    By inversion on \rn{C-Alloc}, we have $G  \mid  \langmv{L'}  \vdash   \overline{ C }   \Rightarrow  \langmv{L}$.
    Now $ \langnt{V}  \langsym{(}  G  \langsym{)}  \setminus   \mathrm{Used}( \langmv{H} )   = \langmv{L'}  \langsym{,}  \langmv{l}$ implies $ \langnt{V}  \langsym{(}  G  \langsym{)}  \setminus   \mathrm{Used}(  \langmv{H} \{ \langmv{x'}  \mapsto  \langmv{l} \}  )   = \langmv{L'}$,
    and thus $G  \mid   \langnt{V}  \langsym{(}  G  \langsym{)}  \setminus   \mathrm{Used}(  \langmv{H} \{ \langmv{x'}  \mapsto  \langmv{l} \}  )    \vdash   \overline{ C }   \Rightarrow  \langmv{L}$.
    Therefore, we have $\Theta \mid \Gamma  \langsym{,}  \langmv{x'}  \langsym{:}   \texttt{qbit}( \langmv{l} )  \mid \langmv{D} \mid G \vdash \rstate{ H\{ x' \mapsto l \} }{ \rho \otimes \lblketbra{0}{0}{l} }{ \langsym{[}  \langmv{x'}  \slash  \langmv{x}  \langsym{]} \, \langnt{e} }$.

  \pfcaselabel{\rn{T-Free}}\\
    Similar to the case of \rn{T-Init}.

  \pfcaselabel{\rn{E-While}}
    \begin{gather*}
      \Theta  \mid  \Gamma  \mid  \langmv{D}  \mid  G  \vdash   \left[  \langmv{H} ,  \rho ,  \langkw{while} \, \langnt{e_{{\mathrm{1}}}} \, \langkw{do} \, \langnt{e_{{\mathrm{2}}}}  \right]   \Rightarrow  \Gamma  \mid  \langmv{L} \\
       \left[  \langmv{H} ,  \rho ,  \langkw{while} \, \langnt{e_{{\mathrm{1}}}} \, \langkw{do} \, \langnt{e_{{\mathrm{2}}}}  \right]  \\
      \qquad \evalarrow  \left[  \langmv{H} ,  \rho ,  \langkw{if} \, \langnt{e_{{\mathrm{1}}}} \, \langkw{then} \, \langnt{e_{{\mathrm{2}}}}  \langsym{;}  \langsym{(}  \langkw{while} \, \langnt{e_{{\mathrm{1}}}} \, \langkw{do} \, \langnt{e_{{\mathrm{2}}}}  \langsym{)} \, \langkw{else} \, \langsym{(}  \langsym{)}  \right] 
    \end{gather*}
    By inversion on \rn{T-RState}, we have
    \begin{gather*}
      \Theta  \mid  \Gamma  \vdash  \langkw{while} \, \langnt{e_{{\mathrm{1}}}} \, \langkw{do} \, \langnt{e_{{\mathrm{2}}}}  \langsym{:}  \langkw{unit}  \Rightarrow  \Gamma  \mid   \langsym{(}   \overline{ C_{{\mathrm{1}}} }   \mathbin{+\mkern-10mu+}   \overline{ C_{{\mathrm{2}}} }   \langsym{)} ^*   \mathbin{+\mkern-10mu+}   \overline{ C_{{\mathrm{1}}} }  \\
      G  \mid   \langnt{V}  \langsym{(}  G  \langsym{)}  \setminus   \mathrm{Used}( \langmv{H} )    \vdash   \langsym{(}   \overline{ C_{{\mathrm{1}}} }   \mathbin{+\mkern-10mu+}   \overline{ C_{{\mathrm{2}}} }   \langsym{)} ^*   \mathbin{+\mkern-10mu+}   \overline{ C_{{\mathrm{1}}} }   \Rightarrow  \langmv{L}.
    \end{gather*}
    By inversion on \rn{T-While}, we have 
    \begin{gather*}
      \Theta  \mid  \Gamma  \vdash  \langnt{e_{{\mathrm{1}}}}  \langsym{:}   \texttt{bool}   \Rightarrow  \Gamma  \mid   \overline{ C_{{\mathrm{1}}} } 
      \andalso \Theta  \mid  \Gamma  \vdash  \langnt{e_{{\mathrm{2}}}}  \langsym{:}  \langkw{unit}  \Rightarrow  \Gamma  \mid   \overline{ C_{{\mathrm{2}}} } .
    \end{gather*}
    By \rn{T-Seq}, we have $\Theta  \mid  \Gamma  \vdash  \langnt{e_{{\mathrm{2}}}}  \langsym{;}  \langkw{while} \, \langnt{e_{{\mathrm{1}}}} \, \langkw{do} \, \langnt{e_{{\mathrm{2}}}}  \langsym{:}  \langkw{unit}  \Rightarrow  \Gamma  \mid     \overline{ C_{{\mathrm{2}}} }    \mathbin{+\mkern-10mu+}  \langsym{(}    \overline{ C_{{\mathrm{1}}} }    \mathbin{+\mkern-10mu+}   \overline{ C_{{\mathrm{2}}} }   \langsym{)} ^*   \mathbin{+\mkern-10mu+}   \overline{ C_{{\mathrm{1}}} } $.
    Also $\Theta  \mid  \Gamma  \vdash  \langsym{(}  \langsym{)}  \langsym{:}  \langkw{unit}  \Rightarrow  \Gamma  \mid   \epsilon $.
    Therefore, we have $\Theta  \mid  \Gamma  \vdash  \langkw{if} \, \langnt{e_{{\mathrm{1}}}} \, \langkw{then} \, \langnt{e_{{\mathrm{2}}}}  \langsym{;}  \langsym{(}  \langkw{while} \, \langnt{e_{{\mathrm{1}}}} \, \langkw{do} \, \langnt{e_{{\mathrm{2}}}}  \langsym{)} \, \langkw{else} \, \langsym{(}  \langsym{)}  \langsym{:}  \langkw{unit}  \Rightarrow  \Gamma  \mid    \overline{ C_{{\mathrm{1}}} }    \mathbin{+\mkern-10mu+}  \langsym{(}   \langsym{(}     \overline{ C_{{\mathrm{2}}} }    \mathbin{+\mkern-10mu+}  \langsym{(}    \overline{ C_{{\mathrm{1}}} }    \mathbin{+\mkern-10mu+}   \overline{ C_{{\mathrm{2}}} }   \langsym{)} ^*   \mathbin{+\mkern-10mu+}   \overline{ C_{{\mathrm{1}}} }   \langsym{)}  \lor   \epsilon    \langsym{)}$.
    As there are no changes to the environment or the heap, our remaining task is to show $G  \mid   \langnt{V}  \langsym{(}  G  \langsym{)}  \setminus   \mathrm{Used}( \langmv{H} )    \vdash    \overline{ C_{{\mathrm{1}}} }    \mathbin{+\mkern-10mu+}  \langsym{(}   \langsym{(}     \overline{ C_{{\mathrm{2}}} }    \mathbin{+\mkern-10mu+}  \langsym{(}    \overline{ C_{{\mathrm{1}}} }    \mathbin{+\mkern-10mu+}   \overline{ C_{{\mathrm{2}}} }   \langsym{)} ^*   \mathbin{+\mkern-10mu+}   \overline{ C_{{\mathrm{1}}} }   \langsym{)}  \lor   \epsilon    \langsym{)}  \Rightarrow  \langmv{L}$ for some $\langmv{L'}$.
    We prove this by showing $G  \mid   \langnt{V}  \langsym{(}  G  \langsym{)}  \setminus   \mathrm{Used}( \langmv{H} )    \vdash    \overline{ C_{{\mathrm{1}}} }   \mathbin{+\mkern-10mu+}   \overline{ C_{{\mathrm{2}}} }   \mathbin{+\mkern-10mu+}  \langsym{(}    \overline{ C_{{\mathrm{1}}} }    \mathbin{+\mkern-10mu+}   \overline{ C_{{\mathrm{2}}} }   \langsym{)} ^*   \mathbin{+\mkern-10mu+}   \overline{ C_{{\mathrm{1}}} }   \Rightarrow  \langmv{L}$ and $G  \mid   \langnt{V}  \langsym{(}  G  \langsym{)}  \setminus   \mathrm{Used}( \langmv{H} )    \vdash   \overline{ C_{{\mathrm{1}}} }   \Rightarrow  \langmv{L}$ because of \cref{lem:C-branch-distribute}.

    By applying \cref{lem:C-remove-loop} to $G  \mid   \langnt{V}  \langsym{(}  G  \langsym{)}  \setminus   \mathrm{Used}( \langmv{H} )    \vdash   \langsym{(}   \overline{ C_{{\mathrm{1}}} }   \mathbin{+\mkern-10mu+}   \overline{ C_{{\mathrm{2}}} }   \langsym{)} ^*   \mathbin{+\mkern-10mu+}   \overline{ C_{{\mathrm{1}}} }   \Rightarrow  \langmv{L}$, we have $G  \mid   \langnt{V}  \langsym{(}  G  \langsym{)}  \setminus   \mathrm{Used}( \langmv{H} )    \vdash   \overline{ C_{{\mathrm{1}}} }   \Rightarrow  \langmv{L}$.
    Next, by \cref{lem:C-loop-fold}, we can prove $G  \mid   \langnt{V}  \langsym{(}  G  \langsym{)}  \setminus   \mathrm{Used}( \langmv{H} )    \vdash    \overline{ C_{{\mathrm{1}}} }   \mathbin{+\mkern-10mu+}   \overline{ C_{{\mathrm{2}}} }   \mathbin{+\mkern-10mu+}  \langsym{(}   \overline{ C_{{\mathrm{1}}} }   \mathbin{+\mkern-10mu+}   \overline{ C_{{\mathrm{2}}} }   \langsym{)} ^*   \Rightarrow  \langmv{L}$.
    Therefore we have $G  \mid   \langnt{V}  \langsym{(}  G  \langsym{)}  \setminus   \mathrm{Used}( \langmv{H} )    \vdash    \overline{ C_{{\mathrm{1}}} }   \mathbin{+\mkern-10mu+}   \overline{ C_{{\mathrm{2}}} }   \mathbin{+\mkern-10mu+}  \langsym{(}   \overline{ C_{{\mathrm{1}}} }   \mathbin{+\mkern-10mu+}   \overline{ C_{{\mathrm{2}}} }   \langsym{)} ^*   \mathbin{+\mkern-10mu+}   \overline{ C_{{\mathrm{1}}} }   \Rightarrow  \langmv{L}$ by \rn{C-Concat}.
    Now we showed both statements and thus we have $G  \mid   \langnt{V}  \langsym{(}  G  \langsym{)}  \setminus   \mathrm{Used}( \langmv{H} )    \vdash    \overline{ C_{{\mathrm{1}}} }    \mathbin{+\mkern-10mu+}  \langsym{(}   \langsym{(}     \overline{ C_{{\mathrm{2}}} }    \mathbin{+\mkern-10mu+}  \langsym{(}    \overline{ C_{{\mathrm{1}}} }    \mathbin{+\mkern-10mu+}   \overline{ C_{{\mathrm{2}}} }   \langsym{)} ^*   \mathbin{+\mkern-10mu+}   \overline{ C_{{\mathrm{1}}} }   \langsym{)}  \lor   \epsilon    \langsym{)}  \Rightarrow  \langmv{L}$ by \cref{lem:C-branch-distribute}.

  \pfcaselabel{\rn{E-Seq1}}
    \begin{gather*}
      \Theta  \mid  \Gamma  \mid  \langmv{D}  \mid  G  \vdash   \left[  \langmv{H} ,  \rho ,  \langnt{e_{{\mathrm{1}}}}  \langsym{;}  \langnt{e_{{\mathrm{2}}}}  \right]   \Rightarrow  \Gamma'  \mid  \langmv{L} \\
       \left[  \langmv{H} ,  \rho ,  \langnt{e_{{\mathrm{1}}}}  \langsym{;}  \langnt{e_{{\mathrm{2}}}}  \right] 
      \evalarrow  \left[  \langmv{H'} ,  \rho' ,  \langnt{e'_{{\mathrm{1}}}}  \langsym{;}  \langnt{e_{{\mathrm{2}}}}  \right]  \\
       \left[  \langmv{H} ,  \rho ,  \langnt{e_{{\mathrm{1}}}}  \right] 
      \evalarrow  \left[  \langmv{H'} ,  \rho' ,  \langnt{e'_{{\mathrm{1}}}}  \right] 
    \end{gather*}
    By inversion on \rn{T-RState}, we have
    \begin{gather*}
      \Theta  \mid  \Gamma  \vdash  \langnt{e_{{\mathrm{1}}}}  \langsym{;}  \langnt{e_{{\mathrm{2}}}}  \langsym{:}  \tau  \Rightarrow  \Gamma'  \mid   \overline{ C } 
      \andalso G  \mid   \langnt{V}  \langsym{(}  G  \langsym{)}  \setminus   \mathrm{Used}( \langmv{H} )    \vdash   \overline{ C }   \Rightarrow  \langmv{L} \\
      \dom(\Gamma) \subseteq \dom(H)
      \andalso  \mathrm{Used}( \Gamma )  \subseteq  \mathrm{Used}( \langmv{H} ) .
    \end{gather*}
    By inversion on \rn{T-Seq}, there exists $\Gamma_{{\mathrm{1}}}$, $ \overline{ C_{{\mathrm{1}}} } $ and $ \overline{ C_{{\mathrm{2}}} } $ such that
    \begin{gather*}
      \Theta  \mid  \Gamma  \vdash  \langnt{e_{{\mathrm{1}}}}  \langsym{:}  \langkw{unit}  \Rightarrow  \Gamma_{{\mathrm{1}}}  \mid   \overline{ C_{{\mathrm{1}}} } 
      \andalso \Theta  \mid  \Gamma_{{\mathrm{1}}}  \vdash  \langnt{e_{{\mathrm{2}}}}  \langsym{:}  \tau  \Rightarrow  \Gamma'  \mid   \overline{ C_{{\mathrm{2}}} } 
      \andalso  \overline{ C }  =  \overline{ C_{{\mathrm{1}}} }   \mathbin{+\mkern-10mu+}   \overline{ C_{{\mathrm{2}}} } .
    \end{gather*}
    By \cref{lem:C-inversion-concat-arbitrary}, we have
    \begin{gather*}
      G  \mid   \langnt{V}  \langsym{(}  G  \langsym{)}  \setminus   \mathrm{Used}( \langmv{H} )    \vdash   \overline{ C_{{\mathrm{1}}} }   \Rightarrow  \langmv{L_{{\mathrm{1}}}}
      \andalso G  \mid  \langmv{L_{{\mathrm{1}}}}  \vdash   \overline{ C_{{\mathrm{2}}} }   \Rightarrow  \langmv{L}
    \end{gather*}
    for some $\langmv{L_{{\mathrm{1}}}}$.
    By the induction hypothesis, we have $\Theta  \mid  \Gamma''  \mid  \langmv{D}  \mid  G  \vdash   \left[  \langmv{H} ,  \rho ,  \langnt{e'_{{\mathrm{1}}}}  \right]   \Rightarrow  \Gamma_{{\mathrm{1}}}  \mid  \langmv{L_{{\mathrm{1}}}}$ for some $\Gamma''$.
    Thus, by inversion on \rn{T-RState}, we have
    \begin{gather*}
      \Theta  \mid  \Gamma''  \vdash  \langnt{e'_{{\mathrm{1}}}}  \langsym{:}  \langkw{unit}  \Rightarrow  \Gamma_{{\mathrm{1}}}  \mid   \overline{ C'_{{\mathrm{1}}} } 
      \andalso G  \mid   \langnt{V}  \langsym{(}  G  \langsym{)}  \setminus   \mathrm{Used}( \langmv{H'} )    \vdash   \overline{ C'_{{\mathrm{1}}} }   \Rightarrow  \langmv{L_{{\mathrm{1}}}} \\
      \dom(\Gamma'') \subseteq \dom(H')
      \andalso  \mathrm{Used}( \Gamma'' )  \subseteq  \mathrm{Used}( \langmv{H'} ) 
    \end{gather*}
    for some $ \overline{ C'_{{\mathrm{1}}} } $.
    By \rn{C-Concat}, we have $G  \mid   \langnt{V}  \langsym{(}  G  \langsym{)}  \setminus   \mathrm{Used}( \langmv{H'} )    \vdash   \overline{ C'_{{\mathrm{1}}} }   \mathbin{+\mkern-10mu+}   \overline{ C_{{\mathrm{2}}} }   \Rightarrow  \langmv{L}$.
    We also have $\Theta  \mid  \Gamma''  \vdash  \langnt{e'_{{\mathrm{1}}}}  \langsym{;}  \langnt{e_{{\mathrm{2}}}}  \langsym{:}  \tau  \Rightarrow  \Gamma'  \mid   \overline{ C'_{{\mathrm{1}}} }   \mathbin{+\mkern-10mu+}   \overline{ C_{{\mathrm{2}}} } $ by \rn{T-Seq}.
    Therefore $\Theta  \mid  \Gamma''  \mid  \langmv{D}  \mid  G  \vdash   \left[  \langmv{H'} ,  \rho' ,  \langnt{e'_{{\mathrm{1}}}}  \langsym{;}  \langnt{e_{{\mathrm{2}}}}  \right]   \Rightarrow  \Gamma'  \mid  \langmv{L}$.

  \pfcaselabel{\rn{E-Call}}
    \begin{gather*}
      \Theta  \mid  \Gamma  \mid  \langmv{D}  \mid  G  \vdash   \left[  \langmv{H} ,  \rho ,  \langmv{f}  \langsym{[}   \overline{ \langmv{l} }   \langsym{]}  \langsym{(}  \langmv{x_{{\mathrm{1}}}}  \langsym{,} \, .. \, \langsym{,}  \langmv{x_{\langmv{n}}}  \langsym{)}  \right]   \Rightarrow  \Gamma'  \mid  \langmv{L'} \\
       \langmv{f}  \mapsto [   \overline{ \langmv{l'} }   ]( \langmv{x'_{{\mathrm{1}}}}  \langsym{,} \, .. \, \langsym{,}  \langmv{x'_{\langmv{n}}} ) \langnt{e}  \in \langmv{D} \\
       \left[  \langmv{H} ,  \rho ,  \langmv{f}  \langsym{[}   \overline{ \langmv{l} }   \langsym{]}  \langsym{(}  \langmv{x_{{\mathrm{1}}}}  \langsym{,} \, .. \, \langsym{,}  \langmv{x_{\langmv{n}}}  \langsym{)}  \right] 
      \evalarrow
       \left[  \langmv{H} ,  \rho ,   \sigma_{ \langmv{l} }  \,  \sigma_{ \langmv{x} }  \, \langnt{e}  \right]  \\
       \sigma_{ \langmv{l} }  = \langsym{[}   \overline{ \langmv{l} }   \slash   \overline{ \langmv{l'} }   \langsym{]}
      \andalso  \sigma_{ \langmv{x} }  = \langsym{[}  \langmv{x_{{\mathrm{1}}}}  \slash  \langmv{x'_{{\mathrm{1}}}}  \langsym{]} \, .. \, \langsym{[}  \langmv{x_{\langmv{n}}}  \slash  \langmv{x'_{\langmv{n}}}  \langsym{]}
    \end{gather*}
    By inversion on \rn{T-RState} and \rn{T-Call}, we have
    \begin{gather*}
      \Theta  \langsym{(}  \langmv{f}  \langsym{)} =  \Pi   \overline{ \langmv{l'} }   . \braket{ \langmv{x'_{{\mathrm{1}}}}  \langsym{:}  \tau_{{\mathrm{1}}}  \langsym{,} \, .. \, \langsym{,}  \langmv{x'_{\langmv{n}}}  \langsym{:}  \tau_{\langmv{n}} } \xrightarrow{   \overline{ C }   } \braket{ \Gamma''  |  \tau }  
      \andalso \Gamma' = \Gamma'''  \langsym{,}   \sigma_{ \langmv{l} }  \,  \sigma_{ \langmv{x} }  \, \Gamma''\\
      \Theta  \mid  \Gamma  \vdash  \langmv{f}  \langsym{[}   \overline{ \langmv{l} }   \langsym{]}  \langsym{(}  \langmv{x_{{\mathrm{1}}}}  \langsym{,} \, .. \, \langsym{,}  \langmv{x_{\langmv{n}}}  \langsym{)}  \langsym{:}   \sigma_{ \langmv{l} }  \, \tau  \Rightarrow  \Gamma'  \mid   \sigma_{ \langmv{l} }  \,  \overline{ C }  \\
      \Gamma = \Gamma'''  \langsym{,}  \langmv{L}  \langsym{,}  \langmv{x_{{\mathrm{1}}}}  \langsym{:}   \sigma_{ \langmv{l} }  \, \tau_{{\mathrm{1}}}  \langsym{,} \, .. \, \langsym{,}  \langmv{x_{\langmv{n}}}  \langsym{:}   \sigma_{ \langmv{l} }  \, \tau_{\langmv{n}}
      \andalso \langmv{L} =  \overline{ \langmv{l} }  \setminus (\cup_{i=1}^n  \mathrm{flv}(  \sigma_{ \langmv{l} }  \, \tau_{\langmv{i}} ) ) \\
      \dom(\Gamma) \subseteq \dom(H)
      \andalso  \mathrm{Used}( \Gamma )  \subseteq  \mathrm{Used}( \langmv{H} ) 
      \andalso G  \mid   \mathrm{Used}( \langmv{H} )   \vdash   \sigma_{ \langmv{l} }  \,  \overline{ C }   \Rightarrow  \langmv{L'}.
    \end{gather*}
    Moreover, by \cref{lem:inv-ty-fundecl}, we have 
    \begin{gather*}
      \Theta' \subseteq \Theta
      \andalso  \overline{ \langmv{l'} }  = \langmv{L''} \uplus (\bigcup_{i=1}^n  \mathrm{flv}( \tau_{\langmv{i}} ) ) \\
      \Theta'  \mid  \langmv{L''}  \langsym{,}  \langmv{x'_{{\mathrm{1}}}}  \langsym{:}  \tau_{{\mathrm{1}}}  \langsym{,} \, .. \, \langsym{,}  \langmv{x'_{\langmv{n}}}  \langsym{:}  \tau_{\langmv{n}}  \vdash  \langnt{e}  \langsym{:}  \tau  \Rightarrow  \Gamma''  \mid   \overline{ C } .
    \end{gather*}
    By \cref{lem:substitution-var,lem:weakening-ftyenv}, we have $\Theta  \mid  \langmv{L''}  \langsym{,}  \langmv{x_{{\mathrm{1}}}}  \langsym{:}  \tau_{{\mathrm{1}}}  \langsym{,} \, .. \, \langsym{,}  \langmv{x_{\langmv{n}}}  \langsym{:}  \tau_{\langmv{n}}  \vdash   \sigma_{ \langmv{x} }  \, \langnt{e}  \langsym{:}  \tau  \Rightarrow   \sigma_{ \langmv{x} }  \, \Gamma''  \mid   \overline{ C } $.
    Location variables in $ \overline{ \langmv{l} } $ are distinct and thus $ \sigma_{ \langmv{l} }  \,  \overline{ \langmv{l'} }  =  \overline{ \langmv{l} }  =  \sigma_{ \langmv{l} }  \, \langmv{L''} \uplus (\cup_{i=1}^n  \mathrm{flv}(  \sigma_{ \langmv{l} }  \, \tau_{\langmv{i}} ) )$.
    This and $\langmv{L} =  \overline{ \langmv{l} }  \setminus (\cup_{i=1}^n  \mathrm{flv}(  \sigma_{ \langmv{l} }  \, \tau_{\langmv{i}} ) )$ imply $ \sigma_{ \langmv{l} }  \, \langmv{L''} = \langmv{L}$.
    Therefore $\Theta  \mid  \langmv{L}  \langsym{,}  \langmv{x_{{\mathrm{1}}}}  \langsym{:}   \sigma_{ \langmv{l} }  \, \tau_{{\mathrm{1}}}  \langsym{,} \, .. \, \langsym{,}  \langmv{x_{\langmv{n}}}  \langsym{:}   \sigma_{ \langmv{l} }  \, \tau_{\langmv{n}}  \vdash   \sigma_{ \langmv{l} }  \,  \sigma_{ \langmv{x} }  \, \langnt{e}  \langsym{:}   \sigma_{ \langmv{l} }  \, \tau  \Rightarrow   \sigma_{ \langmv{l} }  \,  \sigma_{ \langmv{x} }  \, \Gamma''  \mid   \sigma_{ \langmv{l} }  \,  \overline{ C } $ by \cref{lem:substitution-locs}.
    Then we have $\Theta  \mid  \Gamma  \vdash   \sigma_{ \langmv{l} }  \,  \sigma_{ \langmv{x} }  \, \langnt{e}  \langsym{:}   \sigma_{ \langmv{l} }  \, \tau  \Rightarrow  \Gamma'  \mid   \sigma_{ \langmv{l} }  \,  \overline{ C } $ by \cref{lem:weakening-tyenv}.
    Finally we can have $\Theta  \mid  \Gamma  \mid  \langmv{D}  \mid  G  \vdash   \left[  \langmv{H} ,  \rho ,   \sigma_{ \langmv{l} }  \,  \sigma_{ \langmv{x} }  \, \langnt{e}  \right]   \Rightarrow  \Gamma'  \mid  \langmv{L''}$ because other premises of \rn{T-RState} hold obviously.

  Other cases can be shown in the same way as in the previous cases.
\end{proof}

Now we can prove the type soundness of $\langname$ (\cref{thm:soundness}) by the progress and type preservation lemmas.
\begin{proof}[\Cref{thm:soundness}]
  By standard progress (\cref{lem:progress}) and subject reduction lemma (\cref{lem:type-preserve}).
\end{proof}

To conclude this section, we will prove \cref{prop:ser-sound,prop:ser-length} about serialized command sequences, which are mentioned in the paper.

\begin{proof}[\Cref{prop:ser-sound}]
  Prove by induction on the size of $\overline{C}$.
  \pfcase{$\overline{C} =  \epsilon $}
    Trivial.
  \pfcase{$\overline{C} =  \langmv{l_{{\mathrm{1}}}}  \sim  \langmv{l_{{\mathrm{2}}}}   \mathbin{+\mkern-10mu+}   \overline{ C' } $}
    By the definition, we have $ \mathrm{ser}( \overline{C} )  =  \langmv{l_{{\mathrm{1}}}}  \sim  \langmv{l_{{\mathrm{2}}}}   \mathbin{+\mkern-10mu+}   \mathrm{ser}(  \overline{ C' }  ) $.
    If $G  \mid  \langmv{L}  \vdash   \langmv{l_{{\mathrm{1}}}}  \sim  \langmv{l_{{\mathrm{2}}}}   \mathbin{+\mkern-10mu+}   \mathrm{ser}(  \overline{ C' }  )   \Rightarrow  \langmv{L}$, then we have $G  \mid  \langmv{L}  \vdash   \langmv{l_{{\mathrm{1}}}}  \sim  \langmv{l_{{\mathrm{2}}}}   \Rightarrow  \langmv{L}$ and $G  \mid  \langmv{L}  \vdash   \mathrm{ser}(  \overline{ C' }  )   \Rightarrow  \langmv{L}$ by the inversion on \rn{C-Concat}.
    By the induction hypothesis, $G  \mid  \langmv{L}  \vdash   \mathrm{ser}(  \overline{ C' }  )   \Rightarrow  \langmv{L} \Leftrightarrow G  \mid  \langmv{L}  \vdash   \overline{ C' }   \Rightarrow  \langmv{L'}$ for some $\langmv{L'}$.
    Therefore we have $G  \mid  \langmv{L}  \vdash   \langmv{l_{{\mathrm{1}}}}  \sim  \langmv{l_{{\mathrm{2}}}}   \mathbin{+\mkern-10mu+}   \overline{ C' }   \Rightarrow  \langmv{L'}$ by applying \rn{C-Concat}.
    Similarly we can prove $G  \mid  \langmv{L}  \vdash   \langmv{l_{{\mathrm{1}}}}  \sim  \langmv{l_{{\mathrm{2}}}}   \mathbin{+\mkern-10mu+}   \overline{ C' }   \Rightarrow  \langmv{L'} \Rightarrow G  \mid  \langmv{L}  \vdash   \langmv{l_{{\mathrm{1}}}}  \sim  \langmv{l_{{\mathrm{2}}}}   \mathbin{+\mkern-10mu+}   \mathrm{ser}(  \overline{ C' }  )   \Rightarrow   \emptyset $.
  \pfcase{$\overline{C} = \langkw{alloc} \, \langsym{(}  \langmv{l''}  \langsym{)}  \mathbin{+\mkern-10mu+}   \overline{ C' } $}
    By the definition, $ \mathrm{ser}( \overline{C} )  = \langkw{alloc} \, \langsym{(}  \langmv{l''}  \langsym{)}  \mathbin{+\mkern-10mu+}   \mathrm{ser}(  \overline{ C' }  )   \mathbin{+\mkern-10mu+}  \langkw{free} \, \langsym{(}  \langmv{l''}  \langsym{)}$.
    If $G  \mid  \langmv{L}  \vdash  \langkw{alloc} \, \langsym{(}  \langmv{l''}  \langsym{)}  \mathbin{+\mkern-10mu+}   \mathrm{ser}(  \overline{ C' }  )   \mathbin{+\mkern-10mu+}  \langkw{free} \, \langsym{(}  \langmv{l''}  \langsym{)}  \Rightarrow  \langmv{L}$, then we have $G  \mid  \langmv{L}  \vdash  \langkw{alloc} \, \langsym{(}  \langmv{l}  \langsym{)}  \Rightarrow   \langmv{L}  \setminus  \langmv{l''} $ and $G  \mid   \langmv{L}  \setminus  \langmv{l''}   \vdash   \mathrm{ser}(  \overline{ C' }  )   \Rightarrow   \langmv{L}  \setminus  \langmv{l''} $ by the inversion on \rn{C-Concat}, \rn{C-Alloc}. 
    Then we have $G  \mid   \langmv{L}  \setminus  \langmv{l''}   \vdash   \overline{ C' }   \Rightarrow  \langmv{L'''}$ for some $\langmv{L'''}$ by the induction hypothesis.
    Therefore, we have $G  \mid  \langmv{L}  \vdash  \langkw{alloc} \, \langsym{(}  \langmv{l}  \langsym{)}  \mathbin{+\mkern-10mu+}   \overline{ C' }   \Rightarrow  \langmv{L'''}$ by using \rn{C-Concat}.

    Next we assume that $G  \mid  \langmv{L}  \vdash  \langkw{alloc} \, \langsym{(}  \langmv{l''}  \langsym{)}  \mathbin{+\mkern-10mu+}   \overline{ C' }   \Rightarrow  \langmv{L'}$.
    By the inversion on \rn{C-Concat}, we have $G  \mid  \langmv{L}  \vdash  \langkw{alloc} \, \langsym{(}  \langmv{l}  \langsym{)}  \Rightarrow   \langmv{L}  \setminus  \langmv{l''} $ and $G  \mid   \langmv{L}  \setminus  \langmv{l''}   \vdash   \overline{ C' }   \Rightarrow  \langmv{L'}$.
    Thus we have $G  \mid   \langmv{L}  \setminus  \langmv{l''}   \vdash   \mathrm{ser}(  \overline{ C' }  )   \Rightarrow   \langmv{L}  \setminus  \langmv{l''} $ by the induction hypothesis.
    Then we can obtain $G  \mid   \langmv{L}  \setminus  \langmv{l''}   \vdash  \langkw{free} \, \langsym{(}  \langmv{l''}  \langsym{)}  \Rightarrow  \langmv{L}$ by just using \rn{C-Free}.
    Finally by \rn{T-Concat}, we have $G  \mid  \langmv{L}  \vdash  \langkw{alloc} \, \langsym{(}  \langmv{l''}  \langsym{)}  \mathbin{+\mkern-10mu+}   \mathrm{ser}(  \overline{ C' }  )   \mathbin{+\mkern-10mu+}  \langkw{free} \, \langsym{(}  \langmv{l''}  \langsym{)}  \Rightarrow  \langmv{L}$.
  \pfcase{$\overline{C} = \langkw{free} \, \langsym{(}  \langmv{l''}  \langsym{)}  \mathbin{+\mkern-10mu+}   \overline{ C' } $}
    Almost the same as the case of allocation.
  \pfcase{$\overline{C} =   \overline{ C_{{\mathrm{1}}} }  ^*   \mathbin{+\mkern-10mu+}   \overline{ C_{{\mathrm{2}}} } $}
    By the definition, we have $ \mathrm{ser}( \overline{C} )  =  \mathrm{ser}(  \overline{ C_{{\mathrm{1}}} }   \mathbin{+\mkern-10mu+}   \overline{ C_{{\mathrm{2}}} }  ) $.
    \begin{align*}
       &\ G  \mid  \langmv{L}  \vdash   \mathrm{ser}(  \overline{ C_{{\mathrm{1}}} }   \mathbin{+\mkern-10mu+}   \overline{ C_{{\mathrm{2}}} }  )   \Rightarrow  \langmv{L} & \\
      \Leftrightarrow &\ G  \mid  \langmv{L}  \vdash   \overline{ C_{{\mathrm{1}}} }   \mathbin{+\mkern-10mu+}   \overline{ C_{{\mathrm{2}}} }   \Rightarrow  \langmv{L''} & (\text{by the induction hypothesis}) \\
      \Leftrightarrow &\ G  \mid  \langmv{L}  \vdash   \overline{ C_{{\mathrm{1}}} }   \Rightarrow  \langmv{L} \land G  \mid  \langmv{L}  \vdash   \overline{ C_{{\mathrm{2}}} }   \Rightarrow  \langmv{L''} & (\text{by the inversion on \rn{C-Concat}}) \\
      \Leftrightarrow &\ G  \mid  \langmv{L}  \vdash    \overline{ C_{{\mathrm{1}}} }  ^*   \Rightarrow  \langmv{L} \land G  \mid  \langmv{L}  \vdash   \overline{ C_{{\mathrm{2}}} }   \Rightarrow  \langmv{L''} & (\text{by \rn{C-Loop}}) \\
      \Leftrightarrow &\ G  \mid  \langmv{L}  \vdash    \overline{ C_{{\mathrm{1}}} }  ^*   \mathbin{+\mkern-10mu+}   \overline{ C_{{\mathrm{2}}} }   \Rightarrow  \langmv{L''} & (\text{by \rn{C-Concat}})
    \end{align*}
  \pfcase{$\overline{C} = \langsym{(}     \overline{ C_{{\mathrm{1}}} }    \lor   \overline{ C_{{\mathrm{2}}} }    \langsym{)}  \mathbin{+\mkern-10mu+}   \overline{ C_{{\mathrm{3}}} } $}
    By the definition, $ \mathrm{ser}( \overline{C} )  =  \mathrm{ser}(  \overline{ C_{{\mathrm{1}}} }  )   \mathbin{+\mkern-10mu+}   \mathrm{ser}(  \overline{ C_{{\mathrm{2}}} }   \mathbin{+\mkern-10mu+}   \overline{ C_{{\mathrm{3}}} }  ) $.
    If $G  \mid  \langmv{L}  \vdash   \mathrm{ser}( \overline{C} )   \Rightarrow   \emptyset $, then we have $G  \mid  \langmv{L}  \vdash   \mathrm{ser}(  \overline{ C_{{\mathrm{1}}} }  )   \Rightarrow   \emptyset $ and $G  \mid  \langmv{L}  \vdash   \mathrm{ser}(  \overline{ C_{{\mathrm{2}}} }   \mathbin{+\mkern-10mu+}   \overline{ C_{{\mathrm{3}}} }  )   \Rightarrow   \emptyset $.
    Thus we have $G  \mid  \langmv{L}  \vdash   \overline{ C_{{\mathrm{1}}} }   \Rightarrow  \langmv{L'}$ and $G  \mid  \langmv{L}  \vdash   \overline{ C_{{\mathrm{2}}} }   \mathbin{+\mkern-10mu+}   \overline{ C_{{\mathrm{3}}} }   \Rightarrow  \langmv{L''}$ for some $\langmv{L'}$ and $\langmv{L''}$ by the induction hypothesis.
    Therefore, by \cref{lem:trans-branch}, we have $G  \mid  \langmv{L}  \vdash  \langsym{(}     \overline{ C_{{\mathrm{1}}} }    \lor   \overline{ C_{{\mathrm{2}}} }    \langsym{)}  \mathbin{+\mkern-10mu+}   \overline{ C_{{\mathrm{3}}} }   \Rightarrow  \langmv{L''}$.

    Next we suppose $G  \mid  \langmv{L}  \vdash  \langsym{(}     \overline{ C_{{\mathrm{1}}} }    \lor   \overline{ C_{{\mathrm{2}}} }    \langsym{)}  \mathbin{+\mkern-10mu+}   \overline{ C_{{\mathrm{3}}} }   \Rightarrow  \langmv{L'}$.
    By \cref{lem:trans-branch}, we have $G  \mid  \langmv{L}  \vdash   \overline{ C_{{\mathrm{1}}} }   \Rightarrow  \langmv{L''}$ and $G  \mid  \langmv{L}  \vdash   \overline{ C_{{\mathrm{2}}} }   \mathbin{+\mkern-10mu+}   \overline{ C_{{\mathrm{3}}} }   \Rightarrow  \langmv{L'}$ for some $\langmv{L''}$.
    Thus we have $G  \mid  \langmv{L}  \vdash   \mathrm{ser}(  \overline{ C_{{\mathrm{1}}} }  )   \Rightarrow  \langmv{L}$ and $G  \mid  \langmv{L}  \vdash   \mathrm{ser}(  \overline{ C_{{\mathrm{2}}} }   \mathbin{+\mkern-10mu+}   \overline{ C_{{\mathrm{3}}} }  )   \Rightarrow  \langmv{L}$ by the induction hypothesis.
    Finally we get $G  \mid  \langmv{L}  \vdash   \mathrm{ser}(  \overline{ C_{{\mathrm{1}}} }  )   \mathbin{+\mkern-10mu+}   \mathrm{ser}(  \overline{ C_{{\mathrm{3}}} }   \mathbin{+\mkern-10mu+}   \overline{ C_{{\mathrm{3}}} }  )   \Rightarrow  \langmv{L}$ by using \rn{C-Concat}.
\end{proof}

\begin{proof}[\Cref{prop:ser-length}]
  We can prove by straightforward induction on the size of of $ \mathrm{ser}( \overline{C} ) $.
\end{proof}

\subsection{Type Checking Algorithm}\label{app:algorithms}

\begin{algorithm}
  \caption{Type checking algorithm based on depth-first search.}
  \label{alg:naive-type-check}
  \begin{algorithmic}[1]
    \Require{$G$ is an architecture graph. $ \overline{ C } $ is a command sequence. $L$ is a set of free locations.}
    \Function{CheckNaive}{$\overline{C}, L$}
      \If{$\overline{C} =  \epsilon $}
        \State \Return true
      \ElsIf{$\overline{C} = \langkw{alloc} \, \langsym{(}  \langmv{l}  \langsym{)}  \mathbin{+\mkern-10mu+}   \overline{ C' } $}
        \State \Return $l \in L$ \&\& \Call{CheckNaive}{$ \overline{ C' } , L \setminus \{l\}$}
      \ElsIf{$\overline{C} = \langkw{free} \, \langsym{(}  \langmv{l}  \langsym{)}  \mathbin{+\mkern-10mu+}   \overline{ C' } $}
        \State \Return \Call{CheckNaive}{$ \overline{ C' } , L \cup \{l\}$}
      \ElsIf{$\overline{C} =  \langmv{l_{{\mathrm{1}}}}  \sim  \langmv{l_{{\mathrm{2}}}}   \mathbin{+\mkern-10mu+}   \overline{ C' } $}
        \State \Return \Call{FindPath}{$G, \langmv{l_{{\mathrm{1}}}}, \langmv{l_{{\mathrm{2}}}}, L$} \&\& \Call{CheckNaive}{$ \overline{ C' } $, $L$}
      \ElsIf{$\overline{C} = \langsym{(}    \overline{ C_{{\mathrm{1}}} }   \lor   \overline{ C_{{\mathrm{2}}} }    \langsym{)}  \mathbin{+\mkern-10mu+}   \overline{ C' } $}
        \State \Return \Call{CheckNaive}{$ \overline{ C_{{\mathrm{1}}} } , L$} $\land$ \Call{CheckNaive}{$  \overline{ C_{{\mathrm{2}}} }    \mathbin{+\mkern-10mu+}   \overline{ C' } , L$}
      \ElsIf{$\overline{C} =   \overline{ C_{{\mathrm{1}}} }  ^*   \mathbin{+\mkern-10mu+}   \overline{ C_{{\mathrm{2}}} } $}
        \State \Return \Call{CheckNaive}{$ \overline{ C_{{\mathrm{1}}} }   \mathbin{+\mkern-10mu+}   \overline{ C_{{\mathrm{2}}} } , L$}
      \EndIf
    \EndFunction
  \end{algorithmic}
\end{algorithm}

}

\end{document}